\documentclass[11pt,twoside]{article}
\usepackage{fancyhdr}
\usepackage[colorlinks,citecolor=blue,urlcolor=blue,linkcolor=blue,bookmarks=false]{hyperref}
\usepackage{amsfonts,epsfig,graphicx}
\usepackage{afterpage}
\usepackage{comment}
\usepackage{amsmath,amssymb,amsthm} 
\usepackage{fullpage}
\usepackage[T1]{fontenc} 
\usepackage{epsf} 
\usepackage{graphics} 
\usepackage{amsfonts,amsmath}
\usepackage[sort]{natbib} 
\usepackage{psfrag,xspace}
\usepackage{color,etoolbox}
\usepackage{thmtools}

\def\inprob{\stackrel{p}{\rightarrow}}
\def\indist{\rightsquigarrow}

\def\T{{ \mathrm{\scriptscriptstyle T} }}

\newcommand{\var}{\text{var}}
\newcommand{\cov}{\text{cov}}

\newcommand{\Pb}{\mathbb{P}}
\newcommand{\Pn}{\mathbb{P}_n}

\newcommand{\E}{\mathbb{E}}
\newcommand{\R}{\mathbb{R}}
\newcommand{\IF}{\mathbb{IF}}
\newcommand{\Holder}{\text{H\"{o}lder}}

\DeclareMathOperator*{\argmin}{arg\,min}

\DeclareSymbolFont{bbold}{U}{bbold}{m}{n}
\DeclareSymbolFontAlphabet{\mathbbold}{bbold}
\newcommand{\one}{\mathbbold{1}}
\newtheorem{theorem}{Theorem}
\newtheorem{lemma}{Lemma}

\newtheorem{proposition}{Proposition}

\theoremstyle{definition}
\newtheorem{definition}{Definition}

\newtheorem{example}{Example}
\theoremstyle{remark}

\newtheorem{remark}{Remark}

\begin{document}

\def\spacingset#1{\renewcommand{\baselinestretch}%
{#1}\small\normalsize} \spacingset{1}

\renewcommand\thmcontinues[1]{continued}

\raggedbottom
\allowdisplaybreaks[1]

%%%%%%%%%%%%%%%%%%%%%%%%%%%%%%%%%%%%%%%%%%%
  
  \title{\vspace*{-.6in} 
  {Semiparametric Doubly Robust  Targeted \\ Double Machine Learning: A Review 
  \footnote{The title is somewhat tongue-in-cheek, reflecting a lack of unified terminology in this area.}
  \footnote{This review largely follows material from a workshop I started giving in 2016 (the first was in March 2016 for the UNC Causal Inference Research Group). Slides can be found \href{http://www.ehkennedy.com/uploads/5/8/4/5/58450265/unc_2019_cirg.pdf}{here}.}}
  }
  \author{\\ Edward H. Kennedy \\ 
    Department of Statistics \& Data Science \\
    Carnegie Mellon University \\ \\ 
    \texttt{edward@stat.cmu.edu} 
    \date{}
    }

  \maketitle
  \thispagestyle{empty}

%\begin{abstract}
%\end{abstract}

%\setcounter{page}{0}
%\pagebreak

\vspace*{-.15in}

\tableofcontents

\vspace{.35in}

\section{Introduction}
\label{sec:setup}

In this review we cover the basics of  efficient nonparametric parameter estimation (also called functional estimation), with a focus on parameters that arise in causal inference problems. We review both efficiency bounds (i.e., what is the best possible performance for estimating a given parameter?) and the analysis of particular estimators (i.e., what is this estimator's error, and does it attain the efficiency bound?) under weak assumptions. \\

We consider the standard setup for functional estimation  problems. Namely, we suppose we observe a sample of independent observations $(Z_1,...,Z_n)$ all identically distributed according to some unknown probability distribution $\Pb$, which is assumed to lie in some model (i.e., set of distributions) $\mathcal{P}$. Importantly, the goal is \emph{not} to estimate all of $\Pb$, or even an individual component of $\Pb$ such as a regression or density function. Instead the goal is to estimate some structured combination of components, called a  \emph{target parameter} or \emph{functional} $\psi: \mathcal{P} \mapsto \R^q$.  A functional can be viewed as a map from the model to some space, which we take as the reals for simplicity (we will often focus on $q=1$, since extensions to the multivariate $q \geq 2$ setup are typically straightforward). \\

We will see throughout this review that the special structure of functionals, being combinations of components of $\Pb$, endows this estimation problem with many interesting nuances. For example, fast $\sqrt{n}$ rates of convergence can be achieved in nonparametric models, in stark contrast to the problems of nonparametric regression or density estimation. \\

By now, there are many reviews and resources available on the topic of  parameter estimation in modern flexible nonparametric models, e.g., 
\citet{bickel1993efficient}, \citet{van2000asymptotic}, \citet{van2002semiparametric}, \citet{van2003unified}, \citet{tsiatis2006semiparametric}, \citet{kosorok2008introduction}, \citet{van2011targeted}, \citet{petersen2014causal}, \citet{kennedy2016semiparametric}, \citet{chernozhukov2018double}, \citet{kennedy2018semiparametric}, \citet{diaz2020machine}, and \citet{hines2022demystifying}, among others. In this review, we put special emphasis on minimax-style efficiency bounds, worked examples, and  practical shortcuts for easing derivations. We gloss over most technical details,  in the interest of highlighting     important concepts and providing intuition  for main ideas. \\

\subsection{Notation}

We write counterfactual outcomes as $Y^a$, i.e., the value of  $Y$ that would have been observed had we set $A=a$. At times we subscript expectations and other quantities with the distribution under which they are taken, i.e., $\E_P(Y \mid X=x)$ for an expectation under distribution $P$. When the distribution is clear from context, we sometimes omit subscripts; in general, quantities without subscripts are meant to be taken under some generic $P$ in the model, or else under the true distribution $\Pb$. We denote convergence in distribution by $\indist$ and convergence in probability by $\inprob$. We use standard big-oh and little-oh notation, i.e., $X_n = O_\Pb(r_n)$ means $X_n/r_n$ is bounded in probability and $X_n = o_\Pb(r_n)$ means $X_n/r_n \inprob 0$. To ease notation we sometimes omit arguments for functions of multiple arguments, e.g., $\varphi = \varphi(z;P)$ when the arguments are clear or secondary to the discussion. We use $\Pn$ to denote the empirical measure so that sample averages are written as $\Pn(f) = \Pn\{f(Z)\} = \frac{1}{n} \sum_i f(Z_i)$. For a possibly random function $\widehat{f}$, we similarly write $\Pb(\widehat{f})=\Pb\{\widehat{f}(Z)\} = \int \widehat{f}(z) \ d\Pb(z)$, and we let $\| \widehat{f} \|^2 = \int \widehat{f}(z)^2 \ d\Pb(z)$ denote the squared $L_2(\Pb)$ norm.   \\

\section{Setup: Target Parameters \& Model Assumptions}

\subsection{Examples of Functionals}
\label{sec:functionalexamples}

Here we give a list of examples of functionals, some of which arise from causal inference problems, and some of which do not. Many (but certainly not all) functionals in causal inference and missing data take the form of regression functions averaged over covariate distributions.  \\

\begin{example}[label=ex:reg](Regression function)
Suppose $Z=(X,Y)$.  The regression of $Y$ on $X$ is given by $\psi=\psi(x) = \E(Y \mid X=x)$.
\end{example}

\medskip

\begin{example}[label=ex:ate](Average treatment effect) 
Suppose $Z=(X,A,Y)$ for confounders $X$, treatment $A$, and outcome $Y$. Then under causal positivity, consistency, and no unmeasured confounding assumptions, the mean outcome if all in the population were treated at level $A=a$ is identified as the expected regression function
$$ \psi = \E(Y^a) = \E\{\E(Y \mid X, A=a) \} , $$
where we use the convention $\E(Y \mid X, A=a) \equiv \mu_a(X)$ for $\mu_a(x) = \E(Y \mid X=x,A=a)$, so the outer expectation in the above is over the marginal distribution of $X$. 
Corresponding contrasts are identified as $\E(Y^1-Y^0) = \E\{ \E(Y \mid X, A=1) - \E(Y \mid X, A=0) \}$, for example. 
\end{example}

\medskip

\begin{remark}[Identifying Assumptions] \label{rem:identify}
Since the focus of this work is on statistical estimation and inference, rather than identification, we only briefly mention causal identifying assumptions, as in the previous example.  The discussion in subsequent sections holds for the parameters of interest defined purely in statistical terms, regardless of whether the causal assumptions hold (excepting statistical assumptions like positivity, which we also mostly gloss over). 
\end{remark}

\medskip

\begin{example}(Mean missing outcome)
Suppose $Z=(X,A,AY)$ for covariates $X$, missing indicator $A$, and outcome $Y$ (which is only observed when $A=1$). Then under positivity and missing at random assumptions, the mean outcome in the population is identified as the expected regression function
$$ \psi = \E(Y) = \E\{ \E(Y \mid X,A=1)\} . $$
This is mathematically equivalent to the mean outcome if all were treated at level $A=1$ from the previous example, and so statistical methods are identical.
\end{example}

\medskip

\begin{example}[label=ex:cov](Variance-weighted treatment effect)
Suppose $Z=(X,A,Y)$ for confounders $X$, treatment $A$, and outcome $Y$. Then under causal positivity, consistency, and no unmeasured confounding assumptions, a variance-weighted average treatment effect is identified as
$$ \psi = \E\{ w(X) \E(Y^1-Y^0 \mid X)\} = \frac{ \E\{ \cov(A,Y \mid X)\}}{\E\{ \var(A \mid X) \}}  $$
for weights $w(X) = \var(A \mid X) / \E\{ \var(A \mid X)\}$ \citep{robins2008higher, li2011higher}. 
\end{example}

\medskip

\begin{example}[label=ex:stochastic](Stochastic intervention effect)
Suppose $Z=(X,A,Y)$ for confounders $X$, treatment $A$, and outcome $Y$. Then under positivity, consistency, and no unmeasured confounding assumptions, the mean outcome if treatments were sampled as $A \sim dG(a \mid x)$ for everyone in the population is identified as
$$ \psi = \E\{ \E(Y \mid X, A^*)\} = \int \int \E(Y \mid X=x,A=a) \ dG(a \mid x) \ d\Pb(x) . $$
We refer to \citet{diaz2012population, haneuse2013estimation, young2014identification}, and \citet{kennedy2019nonparametric} for further discussion. 
\end{example}

\medskip

\begin{example}[label=ex:iv](Instrumental variable effects)
Suppose $Z=(X,R,A,Y)$ for confounders $X$, instrumental variable (IV) $R$, treatment $A$, and outcome $Y$. Then under positivity, consistency, IV-unconfoundedness, exclusion, instrumentation, and monotonicity assumptions, the average treatment effect among compliers with $A^{r=1}>A^{r=0}$ is given by
$$ \psi = \E(Y^{a=1} - Y^{a=0} \mid A^{r=1} > A^{r=0}) = \frac{ \E\{ \E(Y \mid X, R=1) - \E(Y \mid X, R=0) \}}{\E \{ \E(A \mid X, R=1) - \E(A \mid X, R=0)\} } . $$
Replacing monotonicity with an effect homogeneity assumption, the same statistical functional instead represents an effect on the treated \citep{ogburn2015doubly}. Under a different effect homogeneity assumption, the related but different ratio estimand
$$ \psi = \E\left\{ \frac{   \E(Y \mid X, R=1) - \E(Y \mid X, R=0) }{  \E(A \mid X, R=1) - \E(A \mid X, R=0) } \right\} $$
identifies the average treatment effect \citep{wang2018bounded}.
\end{example}

\medskip

\begin{example}[label=ex:gformula](Time-varying treatment effects)
Suppose $Z=(X_1,A_1,...,X_t,A_t,...,X_T,A_T,Y)$ for time-varying confounders $X_t$, treatments $A_t$, and final outcome $Y$. Then under consistency, with sequential versions of positivity and no unmeasured confounding assumptions, the mean outcome if all in the population followed treatment sequence $\overline{a}_T=(a_1,...,a_T)$ is identified as
\begin{align*}
\psi = \E(Y^{\overline{a}_T}) &= \int \cdots \int \E(Y \mid \overline{X}_T=\overline{x}_T, \overline{A}_T=\overline{a}_T) \prod_{t=1}^T d\Pb(x_t \mid \overline{x}_{t-1}, \overline{a}_{t-1})  .
\end{align*}
This is known as Robins' g-formula \citep{robins1986new, robins2009estimation, van2003unified}. The projection of this quantity  (as a function of $\overline{a}_T$) onto an approximating \emph{marginal structural model} $g(\overline{a}_T;\beta)$ is given by
$$ \psi = \argmin_{\beta \in \R^p}  \int w(\overline{a}_T) \Big\{ \E(Y^{\overline{a}_T}) - g(\overline{a}_T;\beta) \Big\}^2 \ d\nu(\overline{a}_T) $$
where $\E(Y^{\overline{a}_T})$ is identified via the expression above, $w$ is a specified weight function, and $\nu$ is a dominating measure for the distribution of $\overline{A}_T$. 
\end{example}

\medskip

\begin{example}(Mediation effects)
Suppose $Z=(X,A,M,Y)$ for confounders $X$, treatment  $A$, mediator $M$, and outcome $Y$. Then under positivity, consistency, and no unmeasured confounding assumptions for $(A,M)$, the controlled direct effect of treatment $A$, keeping the mediator fixed at $M=m$, is identified by
$$ \E(Y^{a,m}-Y^{a',m}) = \E\{ \E(Y \mid X,A=a,M=m) - \E(Y \mid X,A=a',M=m) \} . $$
Indirect effects can be identified analogously \citep{pearl2009causality, imai2010identification, tchetgen2012semiparametric}. The natural direct effect of treatment, when the mediator is set to what it would have been under $A=a$, is identified by
$$ \E(Y^{a,M^a}-Y^{a',M^a}) = \E\{ \E(Y \mid X,A=a)\} - \E \int \E(Y \mid X,A=a',M=m) \ d\Pb(m \mid X, A=a)  . $$
with indirect effects again identified similarly. Natural mediation effects require weaker positivity assumptions than controlled effects. 
\end{example}

\medskip

\begin{example}(Treatment effect bounds)
Suppose $Z=(X,A,Y)$ for confounders $X$, treatment $A$, and outcome $Y$. Then under consistency and positivity assumptions, the average treatment effect is bounded within
$$ \psi = [\psi_\ell,\psi_u] = \E\{ \E(Y \mid X,A=a)\} \pm \delta \Pb(A \neq a)  $$
as long as  $|\E(Y^a \mid X,A=a) - \E(Y^a \mid X, A \neq a)| \leq \delta$ (note this is weaker than no unmeasured confounding, which implies $\delta=0$) \citep{richardson2014nonparametric, luedtke2015statistics}. 
\end{example}

\medskip

\begin{example}[label=ex:expdens](Expected density)
Let $Z$ have density $p$. Then the expected density is
$$ \psi = \E\{ p(Z) \} = \int p(z)^2 \ dz . $$
This functional is a staple of the classical functional estimation literature and arises in tuning parameter selection in density estimation \citep{bickel1988estimating, birge1995estimation}.
\end{example}

\medskip

\begin{example}(Entropy)
Let $Z$ have density $p$.  Then the entropy is 
$$ \psi = -\int p(z) \log p(z) \ dz = - \E\{ \log p(Z) \} . $$
\end{example}

\medskip

\begin{example}($f$-divergence)
Let $Z=(A,Y)$ for $A \in \{0,1\}$ a group indicator and $Y$ a random variable with conditional density $p(y \mid a)$. Then the $f$-divergence of $p(y \mid a=1)$ from $p(y \mid a=0)$ is
$$ \psi = \int f\left( \frac{p(y \mid a=1)}{p(y \mid a=0)} \right) p(y \mid a=0) \ dy $$
for a known function $f$. Particular choices of $f$ yield Kullback-Leibler, Hellinger, total variation, and $\chi^2$ distances, for example \citep{kandasamy2015nonparametric}. 
\end{example}

\bigskip

\subsection{Model Assumptions}

Recall from Section \ref{sec:setup} that the distribution $\Pb$ from which we sample is assumed to lie in a model, i.e., set of distributions $\mathcal{P}$.  In this review we focus on nonparametric models;   in Section \ref{sec:effbd} we typically take $\mathcal{P}$ to be the simplest nonparametric model, consisting of all probability distributions on the sample space, while in Section \ref{sec:methods} we introduce models with smoothness or sparsity. We focus on nonparametric rather than semiparametric models mostly for simplicity;  many ideas extend to the more restricted semiparametric  case, and we refer to  \citet{bickel1993efficient},  \citet{van2003unified}, and \citet{tsiatis2006semiparametric} for more details there. As noted in Remark \ref{rem:identify}, we only briefly mention identifying assumptions, despite their importance, since the focus of this review is on the statistical aspects of causal inference, post-identification. \\

\section{Benchmarks: Nonparametric Efficiency Bounds}
\label{sec:effbd}

After having selected an appropriate target parameter $\psi$ matching the scientific question of interest, identifying (or bounding) it under appropriate causal or other assumptions, and laying out a statistical model $\mathcal{P}$ (which in our case will be nonparametric), a next line of business is to understand \emph{lower bounds} or \emph{benchmarks} for estimation error. In other words, how well can we possibly hope to estimate the parameter $\psi$ over the model $\mathcal{P}$? This is important both theoretically, as a fundamental measure of the statistical difficulty of estimating $\psi$, as well as practically, since it helps tell us whether a particular method is optimally efficient, making the best use of the data (if not, one may need to search for better, more efficient methods). Note there are two parts to showing optimality: (i)  that no estimator can do better than some benchmark, and (ii) that a particular estimator does in fact attain that benchmark. Part (i) is discussed in this section, and part (ii) in the next section.   \\

A classic benchmarking or lower bound result for smooth parametric models is the Cramer-Rao bound \citep{casella2001statistical, van2002semiparametric}. In its simplest form, this result states that for smooth parametric models $\mathcal{P} = \{ P_\theta : \theta \in \R \}$  and smooth functionals (i.e., with $P_\theta$ and $\psi(\theta)$ differentiable in $\theta$), the variance of \emph{any} unbiased estimator $\widehat\psi$ must satisfy
\begin{equation} \label{eq:crao}
\var_\theta(\widehat\psi) \geq \frac{\psi'(\theta)^2}{ \var_\theta\{s_\theta(Z) \}} , 
\end{equation}
where $s_\theta(z) = \frac{\partial}{\partial\theta} \log p_\theta(z)$ is the score function, i.e., no unbiased estimator can have smaller variance than the above ratio. A standard way to benchmark estimation error more generally is through minimax lower bounds of the form
\begin{equation} \label{eq:minimax}
 \inf_{\widehat\psi} \sup_{P \in \mathcal{P}} \E_P\Big[ \{ \widehat\psi - \psi(P) \}^2 \Big] \geq R_n . 
 \end{equation}
These kinds of lower bounds say that the risk for estimating $\psi$ (in this case, in terms of worst-case mean squared error), over the model $\mathcal{P}$, cannot be smaller than $R_n$. For example, when $\psi(P)$ is a density or regression function, and $\mathcal{P}$ is the class of all $s$-smooth \Holder{} densities, then $R_n = Cn^{-1/(1+d/2s)}$ \citep{tsybakov2009introduction}.  \\

Indeed the Cramer-Rao bound \eqref{eq:crao}  also acts  as a benchmark in a more general minimax sense. In fact, for smooth parametric models, one can go beyond global lower bounds of the form \eqref{eq:minimax} and say something about more nuanced local minimax behavior. This is illustrated in the following theorem. \\

\begin{theorem}[Theorem 8.11, \citet{van2000asymptotic}] \label{thm:craominimax}
Assume $P_\theta$ is differentiable in quadratic mean at $\theta$ with nonsingular Fisher information $I_\theta = \var_\theta\{ s_\theta(Z)\}$. If $\psi(\theta)$ is differentiable at $\theta$, with $\psi'(\theta)= \frac{\partial}{\partial \theta}\psi(\theta)$, then for any estimator $\widehat\psi$ it follows that
$$ \inf_{\delta>0} \liminf_{n \rightarrow \infty}  \sup_{\|\theta'-\theta\| < \delta} \ n\  \E_{\theta'} \Big[ \{\widehat\psi - \psi(\theta')\}^2 \Big] \ \geq \ \psi'(\theta)  \var_\theta\{s_\theta(Z) \}^{-1} \psi'(\theta)^\T . $$
\end{theorem}

\bigskip

Intuitively, Theorem \ref{thm:craominimax} says the (asymptotic, worst-case) mean squared error   cannot be smaller than 
$ \psi'(\theta)^2 / n \var_\theta\{s_\theta(Z)\} $, 
for any estimator $\widehat\psi$ in a smooth parametric model.  \\

Thus optimality in the above local asymptotic minimax sense is somewhat settled for \emph{smooth parametric models}. However, what if anything does this say about larger semi- or nonparametric models? Can the above Cramer-Rao bounds be exploited to construct lower bound benchmarks there as well? These questions will be answered in the following subsection. \\

\subsection{Parametric Submodels}

The standard way to connect classic Cramer-Rao bounds for parametric models to larger more complicated nonparametric models  is through a technical device called the  \emph{parametric submodel} \citep{stein1956efficient}. We first give a definition, then describe high-level ideas and give some examples. \\

\begin{definition}
A \emph{parametric submodel} is a smooth parametric model $\mathcal{P}_\epsilon = \{ P_\epsilon : \epsilon \in \R\}$ that satisfies
(i) $\mathcal{P}_\epsilon \subseteq \mathcal{P}$, and (ii) $P_{\epsilon=0} = \Pb$. 
\end{definition}

\bigskip

Thus, in words, a parametric submodel is a parametric model that (i) is contained in the larger model $\mathcal{P}$ of interest, and (ii) equals the true distribution at $\epsilon=0$, i.e., contains the truth $\Pb$. It is important to recognize that a parametric submodel is a technical device used to extend theory from parametric to nonparametric nodels, and not a tool for data analysis \citep{tsiatis2006semiparametric}; in particular, to ensure that property (ii) $P_{\epsilon=0}=\Pb$ holds, parametric submodels must depend on the true distribution $\Pb$, which is of course unknown. \\

The high-level idea behind using submodels is that it is never harder to estimate a parameter over a \emph{smaller} model, relative to a larger one in which the smaller model is contained. So any lower bound for a submodel will also be a valid lower bound for the larger model $\mathcal{P}$. Of course, valid but vacuous lower bounds are easy to construct (e.g., the mean squared error can be no less than zero), so in the next section we will also have to show that these bounds are \emph{relevant}, in the sense that they can actually be attained under some plausible conditions. \\

It turns out that, for the purposes of constructing lower bound benchmarks for functional estimation, it often suffices to use one-dimensional parametric submodels.  A common choice of submodel for nonparametric $\mathcal{P}$ is, for some mean-zero function $h:\mathcal{Z} \rightarrow \R$, 
\begin{equation} \label{eq:exsubmodel}
 p_\epsilon(z) = d\Pb(z) \{ 1 + \epsilon h(z) \} 
 \end{equation}
where $\|h \|_\infty \leq M < \infty$ and $\epsilon < 1/M$ so that $p_\epsilon(z) \geq 0$. Note for this submodel the score function is $\frac{\partial}{\partial\epsilon} \log p_\epsilon(z)|_{\epsilon=0} = \frac{\partial}{\partial\epsilon} \log\{1+\epsilon h(z)\}|_{\epsilon=0} = h(z)$. Therefore the Cramer-Rao lower bound for some $P_\epsilon$ in the example one-dimensional submodel $\mathcal{P}_\epsilon$ above is given by
$$ \frac{\psi'(P_\epsilon)^2}{\var_{P_\epsilon}\{s_\epsilon(Z)\}} = \frac{ \left\{ \frac{\partial}{\partial\epsilon} \psi(P_\epsilon) |_{\epsilon=0} \right\}^2 }{\E_{P_\epsilon}\{h(Z)^2\}} . $$
Other examples of submodels can be found in Section 4.2 of \citet{tsiatis2006semiparametric}, for example. \\

Since any lower bound for the submodel $\mathcal{P}_\epsilon$ is also a lower bound for $\mathcal{P}$, the best and most informative is the \emph{greatest} such lower bound. Can we say anything about the best such lower bound for generic functionals and/or submodels? The next two subsections consider this question. \\

\subsection{Pathwise Differentiability}

Recall the Cramer-Rao bound
\begin{equation} \label{eq:subcrao}
\frac{ \left\{ \frac{\partial}{\partial\epsilon} \psi(P_\epsilon) |_{\epsilon=0} \right\}^2 }{\E_{P_\epsilon}\{s_\epsilon(Z)^2\}} 
\end{equation}
for submodel $\mathcal{P}_\epsilon$ described in the previous subsection. To find the best such lower bound, we would like to optimize the above over all $P_\epsilon$ in some submodel. It is not a priori clear how generally this can be accomplished, since different functionals $\psi$ could yield very different numerators. Therefore let us first consider what we can say about the derivative in the numerator of \eqref{eq:subcrao}, for a large class of \emph{pathwise differentiable} functionals. \\

Namely, suppose the functional $\psi: \mathcal{P} \mapsto \R$ is smooth, as a map from distributions to the reals, in the sense that it admits a kind of \emph{distributional Taylor expansion}
\begin{equation} \label{eq:vonmises}
\psi(\overline{P}) - \psi(P) = \int \varphi(z;\overline{P}) \ d(\overline{P}-P)(z) + R_2(\overline{P},P)
\end{equation}
for distributions $\overline{P}$ and $P$, often called a \emph{von Mises expansion}, where $\varphi(z;P)$ is a mean-zero, finite-variance function satisfying $\int \varphi(z;P) \ dP(z) =0$ and $\int \varphi(z;P)^2 \ dP(z) < \infty$, and $R_2(\overline{P},P)$ is a \emph{second-order remainder} term (which means it only depends on \emph{products} or \emph{squares} of differences between $\overline{P}$ and $P$). \\

Intuitively, the von Mises expansion above is just an infinite-dimensional or distributional analog of a Taylor expansion, with $\varphi(z;Q)$ acting as a usual derivative term; it describes how the functional $\psi$ changes locally when the  distribution changes from $P$ to $\overline{P}$. For example, when $Z\in \{1,...,k\}$ is discrete and so $\overline{P}$ and $P$ have $k$ countable components, the von Mises expansion reduces to a standard multivariate Taylor expansion with
$$  R_2(\overline{P},P) = \psi(\overline{p}_1,...,\overline{p}_k) - \psi(p_1,...,p_k) - \sum_j \frac{\partial}{\partial t_j} \psi(t_1,...,t_k) \Bigm|_{t=\overline{p}} (\overline{p}_j-p_j) . $$
We refer to \citet{fisher2021visually} for more intuition and visual illustrations. \\ 

\begin{remark}
The von Mises terminology comes from, e.g., \citet{vonmises1947asymptotic}, and has been used by \citet{fernholz1983mises}, \citet{van2000asymptotic}, \citet{robins2009quadratic},  \citet{kandasamy2015nonparametric}, among others. The function $\varphi(z;P)$ has been referred to as an influence function, pathwise derivative, gradient, and Neyman orthogonal score \citep{pfanzagl1982contributions, bickel1993efficient, newey1994asymptotic, tsiatis2006semiparametric, van2003unified, chernozhukov2018double}. However it can be important to distinguish between the influence function for a \emph{parameter}, as in \eqref{eq:vonmises}, and the influence function for an \emph{estimator}; this point will be discussed in more detail in the next section. To distinguish between the two, we  typically refer to the influence function for a {parameter} as in \eqref{eq:vonmises} as an \emph{influence curve}. Note that for now, the expansion \eqref{eq:vonmises} is only a smoothness property of the functional $\psi: \mathcal{P} \mapsto \R$, and has nothing to do yet with any data or estimation procedure.  \\
\end{remark}

Many important functionals satisfy the expansion \eqref{eq:vonmises}; we detail a few in the following examples. \\

\begin{example}[continues=ex:ate]
The average treatment effect or missing outcome functional 
$$ \psi(P) = \E_P \{ \E_P(Y \mid X, A=1)\} $$
satisfies \eqref{eq:vonmises} with 
$$ \varphi(Z;P) = \frac{\one(A=1)}{P(A=1 \mid X)} \Big\{ Y - \E_P(Y \mid X, A=1) \Big\} + \E_P(Y \mid X,A=1) - \psi(P) $$
and 
$$ R_2(\overline{P},P) = \int \left\{ \frac{1}{\overline\pi(x)} - \frac{1}{\pi(x)} \right\} \Big\{ \mu(x) - \overline\mu(x) \Big\} \pi(x) \ dP(x) $$
where $\pi(x) = P(A=1 \mid X=x)$ and $\overline\pi(x) = \overline{P}(A=1 \mid X=x)$, and similarly for $\mu(x) = \E_P(Y \mid X=x,A=1)$. 
\end{example}

\bigskip

\begin{example}[continues=ex:cov] The expected conditional covariance functional
$$ \psi(P) = \E_P\{ \cov_P(A,Y \mid X) \}$$
satisfies \eqref{eq:vonmises} with
$$ \varphi(Z;P) = \Big\{ A - \E_P(A \mid X) \Big\} \Big\{ Y - \E_P(Y \mid X) \Big\} - \psi(P) $$
and 
$$ R_2(\overline{P},P) = \int \Big\{ \overline\pi(x) - \pi(x) \Big\} \Big\{ \overline\mu(x) - \mu(x) \Big\}   \ dP(x) $$
where $\pi(x) = \E_P(A \mid X=x)$ and  $\mu(x) = \E_P(Y \mid X=x)$, with $\overline\pi$ and $\overline\mu$ corresponding versions under $\overline{P}$.
\end{example}

\bigskip

\begin{example}[continues=ex:expdens]
The expected density functional
$$ \psi(P) = \E_P\{ p(Z) \} = \int p(z)^2 \ dz $$
satisfies \eqref{eq:vonmises} with
$$ \varphi(Z;P) = 2\Big\{ p(Z) - \psi(P) \Big\} $$
and 
$$ R_2(\overline{P},P) = - \int \Big\{ \overline{p}(z) - p(z) \Big\}^2   \ dz . $$
\end{example}

\bigskip

A related notion of smoothness is \emph{pathwise differentiability}, i.e., that
\begin{equation} \label{eq:pathwise}
\frac{\partial}{\partial \epsilon} \psi(P_\epsilon) \Bigm|_{\epsilon=0} = \int \varphi(z;\Pb) s_\epsilon(z) \ d\Pb(z) 
\end{equation}
for every smooth submodel $P_\epsilon$. This is implied by the von Mises expansion \eqref{eq:vonmises}, under regularity conditions, by taking $(P,Q)=(P_\epsilon,P)$, differentiating both sides and noting that $R_2$ being second order means
$$ \frac{\partial}{\partial \epsilon} R_2(P,P_\epsilon) \Bigm|_{\epsilon=0} = 0 . $$
For more details, see for example Lemma 2 of \citet{kennedy202xdensity}; the above condition is essentially equivalent to what \citet{chernozhukov2018double} refer to as Neyman orthogonality. Pathwise differentiability \eqref{eq:pathwise}, and the von Mises expansion \eqref{eq:vonmises} more generally, play key roles in both deriving lower bound benchmarks (via the \emph{efficient influence function}) and constructing estimators that attain the benchmark. We will continue exploring this first role in the following subsection.   \\

\subsection{The Best Lower Bound  \& Efficient Influence Function}

Armed with the smoothness of our functional $\psi$, as characterized by the von  Mises expansion \eqref{eq:vonmises} and related pathwise differentiability \eqref{eq:pathwise}, we now have enough to characterize the greatest lower bound for generic smooth parametric submodels. \\

For simplicity consider the particular submodel in \eqref{eq:exsubmodel}; it turns out this class of submodel is often sufficient to yield relevant (attainable) lower bounds. For this submodel, the score is $s_\epsilon(z) = h(z)$ and  by pathwise differentiability we have
$$  \frac{\partial}{\partial \epsilon} \psi(P_\epsilon) \Bigm|_{\epsilon=0} = \int \varphi(z;\Pb) h(z) \ d\Pb(z) . $$
Therefore over all Cramer-Rao bounds at $\epsilon=0$ we have
\begin{align*}
\sup_{P_\epsilon } \frac{\psi'(P_\epsilon)^2}{\var\{s_\epsilon(Z)\}} &= \sup_{h} \frac{\E\{ \varphi(Z;\Pb) h(Z) \}^2 }{\E\{h(Z)^2\}} \leq \E\{ \varphi(Z;\Pb)^2 \} = \var\{ \varphi(Z)\}
\end{align*}
where the first equality follows by pathwise differentiability and the form of the submodel, and the inequality by Cauchy-Schwarz. The fact that the greatest lower bound is not just bounded above by $\E(\varphi^2) = \var(\varphi)$, but that this upper bound is actually attained follows from the fact that, for one of the submodels we can take $h(z) = \varphi(z;\Pb)$, as long as $\varphi$ is in the tangent space (i.e., closure of submodel score space). We refer to Lemma 25.19 of \citet{van2000asymptotic} for more details and discussion. \\

\begin{remark}
Recall in this review we are mostly focusing on proper nonparametric models, where the tangent space is the whole Hilbert space of mean-zero, finite-variance functions (see Theorem 4.4 of \citet{tsiatis2006semiparametric}); in that case \eqref{eq:vonmises} only holds for at most one influence curve $\varphi$, which must also be a valid score. However, in proper semiparametric models with a restricted tangent space, the expansion \eqref{eq:vonmises} can hold for potentially many influence curves $\varphi$, and then the one that is also a valid score is called the \emph{efficient influence function}. In contrast, in nonparametric models, there is only one influence curve, and that influence curve is also the efficient influence function. \\
\end{remark}

Therefore the variance of the efficient influence function
\begin{equation} \label{eq:varphi}
\var\{ \varphi(Z;\Pb)\}
\end{equation}
acts as a nonparametric analog of the Cramer-Rao bound. This is critically important, as it implies no estimator can have smaller mean squared error than \eqref{eq:varphi}, in a local asymptotic minimax sense. In particular, if we can show for a particular estimator $\widehat\psi$ that 
$$ \sqrt{n}( \widehat\psi - \psi) \indist N\Big(0, \var\{ \varphi(Z)\} \Big) $$
then we can say the estimator attains the nonparametric efficiency bound. This local asymptotic minimaxity can be formalized as in the following result, for example. \\

\begin{theorem}[Corollary 2.6, \citet{van2002semiparametric}] \label{thm:minimax}
Let $\psi: \mathcal{P} \mapsto \R$ be pathwise differentiable with efficient influence function $\varphi$. Assume the model is nonparametric or the tangent space is a convex cone. Then
$$ \inf_{\delta>0} \liminf_{n \rightarrow \infty}  \sup_{\text{TV}(P,Q)< \delta} \ n\  \E_{Q} \Big[ \{\widehat\psi - \psi(Q)\}^2 \Big] \ \geq \ \var\{ \varphi(Z;P)\} $$
for any estimator sequence $\widehat\psi=\widehat\psi_n$.
\end{theorem}

\bigskip

%\subsection{Aside: Proper Semiparametric Models}

\subsection{Deriving Influence Functions}

In the previous section we showed the crucial importance of the von Mises expansion \eqref{eq:vonmises}, whose derivative term (i.e., influence curve) is the efficient influence function in nonparametric models, and thus the key component of local minimax lower bounds for functional estimation. Further, the efficient influence function not only yields efficiency bounds, but also indicates how to construct efficient estimators and sheds light on the conditions required for such estimators to be efficient (these latter points will be discussed shortly, in the next section). In the last subsection we gave some examples of functionals for which the expansion holds, with particular influence curves and remainder terms, but it may not be clear how to derive these expansions from scratch. This is the main topic of this subsection.  \\

There are several ways to derive influence curves. The most general approach is to explicitly compute the pathwise derivative $\psi'(P_\epsilon)$ for appropriate submodels, set this equal to the right-hand-side of the pathwise differentiability equation \eqref{eq:pathwise}, and solve for the influence curve $\varphi$. In the following example we show how this works in two simple cases, for regression functions with discrete $X$ and the expected density. \\

\begin{example}[continues=ex:reg]
Consider the regression function $\E(Y \mid X=x)$. Here we will show that 
$$ \varphi(z;P)  = \frac{\one(X=x)}{P(X=x)} \Big\{ Y - \E_P(Y \mid X=x) \Big\} $$
is the efficient influence function when $X$ is discrete, by showing that \eqref{eq:pathwise} holds for this choice of $\varphi$.
Let $ s_\epsilon(z) =   \frac{\partial}{\partial\epsilon} \log dP_\epsilon(z)  |_{\epsilon=0} $ denote the submodel score, and note
\begin{align*}
\E\{s(Z) \mid X=x\} &= \int  \frac{\partial}{\partial\epsilon} \log dP_\epsilon(z) \Bigm|_{\epsilon=0} d\Pb(y \mid x) \\
&= \int \frac{\partial}{\partial\epsilon} \log \Big\{ P_\epsilon(X=x) dP_\epsilon(y \mid x) \Big\} \Bigm|_{\epsilon=0}  d\Pb(y \mid x) \\
&= \int \left\{ \frac{\partial}{\partial\epsilon} \log P_\epsilon(X=x) \Bigm|_{\epsilon=0} + \frac{\partial}{\partial\epsilon} \log  dP_\epsilon(y \mid x)  \Bigm|_{\epsilon=0} \right\}  d\Pb(y \mid x) \\
&= \frac{\partial}{\partial\epsilon} \log P_\epsilon(X=x) \Bigm|_{\epsilon=0} 
\end{align*}
where the last equality uses the facts that $\int d\Pb(y \mid x)=1$ and that scores have mean zero, i.e., 
\begin{align*}
\int \frac{\partial}{\partial\epsilon} \log  dP_\epsilon(y \mid x)  \Bigm|_{\epsilon=0} d\Pb(y \mid x) &= \int \frac{\frac{\partial}{\partial\epsilon} dP_\epsilon(y \mid x) \mid_{\epsilon=0} }{d\Pb(y \mid x) } \ d\Pb(y \mid x) \\
&=  \int \frac{\partial}{\partial\epsilon} dP_\epsilon(y \mid x) \Bigm|_{\epsilon=0} =  \frac{\partial}{\partial\epsilon}  \int dP_\epsilon(y \mid x) \Bigm|_{\epsilon=0} =  0 
\end{align*} 
where the first equality used the fact that $\frac{\partial}{\partial \epsilon} \log dP_\epsilon(y \mid x) = \frac{\partial}{\partial \epsilon} dP_\epsilon(y \mid x) / dP_\epsilon(y \mid x)$. Thus in this case the pathwise derivative on the left-hand side of \eqref{eq:pathwise} is
\begin{align*}
\frac{\partial}{\partial\epsilon} \int y \ dP_\epsilon(y \mid x) \Bigm|_{\epsilon=0} &=  \int y \left\{ \frac{\partial}{\partial\epsilon} \log dP_\epsilon(y \mid x) \right\} \Bigm|_{\epsilon=0} d\Pb(y \mid x)  \\
&=  \int y \left\{ \frac{\partial}{\partial\epsilon} \log \frac{dP_\epsilon(z)}{P_\epsilon(X=x)} \right\} \Bigm|_{\epsilon=0} d\Pb(y \mid x)  \\
&=  \int y \left\{ \frac{\partial}{\partial\epsilon} \log dP_\epsilon(z) - \frac{\partial}{\partial\epsilon} \log P_\epsilon(X=x) \right\} \Bigm|_{\epsilon=0}  dP_\epsilon(y \mid x) \\
&= \E\{ Y s_\epsilon(Z) \mid X=x \} - \E\{s_\epsilon(Z) \mid X=x\} \E(Y \mid X=x) .
\end{align*} 
where the first equality holds as long as we can exchange integrals and derivatives and since $P_{\epsilon=0}=\Pb$. 
Now for the right-hand side of \eqref{eq:pathwise} we have
\begin{align*}
\int \varphi(z;\Pb) s_\epsilon(z) \ d\Pb(z) &= \E\left[ \frac{\one(X=x)}{ \Pb(X=x)} \Big\{ Y - \E(Y \mid X=x) \Big\} s_\epsilon(Z)  \right] \\
&=  \E\{Y s_\epsilon(Z) \mid X=x\} - \E\{ s_\epsilon(Z) \mid X=x\} \E(Y \mid X=x) 
\end{align*}
by iterated expectation. This yields the result. 
\end{example}

\bigskip

\begin{example}[continues=ex:expdens]
Let $p_0$ denote the density of $\Pb$. Under regularity conditions, the pathwise derivative for the expected density functional is given by
\begin{align*}
\frac{\partial}{\partial \epsilon} \int p_\epsilon(z)^2 \ dz \Bigm|_{\epsilon=0} &= \int \frac{\partial}{\partial \epsilon} p_\epsilon(z)^2 \ dz \Bigm|_{\epsilon=0} \\
&= \int 2p_\epsilon(z) \frac{\partial}{\partial \epsilon} p_\epsilon(z) \ dz \Bigm|_{\epsilon=0} \\
&= \int 2p_\epsilon(z) \left\{ \frac{\partial}{\partial \epsilon} \log p_\epsilon(z) \right\} p_\epsilon(z) \ dz \Bigm|_{\epsilon=0} \\
&= \int2 \Big\{ p_0(z) - \psi(\Pb) \Big\} \left\{ \frac{\partial}{\partial \epsilon} \log p_\epsilon(z) \right\} \Bigm|_{\epsilon=0}  p_0(z) \ dz 
\end{align*}
where the first equality holds as long as we can exchange integrals and derivatives, the second by the chain rule, the third since $\frac{\partial}{\partial \epsilon} \log p_\epsilon(z) = \frac{\partial}{\partial \epsilon} p_\epsilon(z) / p_\epsilon(z)$, and the last since the score function $s_\epsilon(z) = \frac{\partial}{\partial \epsilon} \log p_\epsilon(z) |_{\epsilon=0}$ has mean zero so that subtracting $\psi$ times the mean does not change the expression.  Now equating the above with the right-hand side of \eqref{eq:pathwise} shows that $2\{p(z) - \psi\}$ is the efficient influence function for the expected density. 
\end{example}

\bigskip

\subsubsection{Two Simple Strategies} 

As was seen above, even for two very simple functionals, the previously described general approach is somewhat indirect and non-constructive. For the mean we  had a putative efficient influence function at our disposal, which may not always be the case, and for the expected density we had to solve an integral equation (which was straightforward in that case, but can be complicated in general). Luckily there are some tricks for making influence function derivations  easier and less time-consuming. We will give two strategies that build off of each other and can both be useful:
\begin{enumerate}
\item Compute Gateaux derivatives assuming data are discrete.
\item Use derivative rules with simple influence functions as building blocks.
\end{enumerate}
The first strategy is somewhat commonplace and has been used and detailed for example in \citet{kandasamy2015nonparametric} and \citet{hines2022demystifying}, for example. We have not seen the second strategy described in the literature. \\

The first step in both strategies is to initially pretend that the data are discrete. This eases calculations and allows for direct computation rather than solving integral equations, while typically still leading to influence functions that are valid in the general continuous or mixed case. The latter can always be verified by checking the general integral version of the pathwise differentiability condition \eqref{eq:pathwise}, for a putative influence function computed by potentially ad hoc means. \citet{ichimura2022influence} show that similar calculations can be used in the  general case by replacing indicators with kernels indexed by a bandwidth converging to zero; however, dealing with indicators eases notation so we use that approach here.   \\

\subsubsection{Strategy 1} 

After reducing to discrete data, the first strategy is to compute the Gateax derivative of the parameter in the direction of a point mass contamination. Specifically, letting $\delta_z=\one(Z=z)$ denote the Dirac measure at $Z=z$, one computes the Gateaux derivative
$$ \frac{\partial}{\partial \epsilon} \psi\{ (1-\epsilon) d\Pb(z) + \epsilon \delta_{z'} \} \Bigm|_{\epsilon=0} $$
which equals the influence function $\varphi(z';P)$. This approach is based on computing the pathwise derivative described in the previous section, but at a special submodel of the form $(1-\epsilon) d\Pb(z) + \epsilon \delta_{z'}$, for which the right-hand side of \eqref{eq:pathwise} happens to equal the influence function itself, rather than its covariance with the score. The reason the latter is true is because the score for this submodel is 
\begin{align*}
\frac{\partial}{\partial \epsilon} \log \Big\{ (1-\epsilon)  d\Pb(z) + \epsilon \delta_{z'} \Big\} \Bigm|_{\epsilon=0} = \frac{\delta_{z'} - d\Pb(z)}{ (1-\epsilon) d\Pb(z) + \epsilon \delta_{z'}}  \Bigm|_{\epsilon=0} = \frac{\delta_{z'}}{d\Pb(z)} - 1
\end{align*}
which implies
\begin{align*}
\int \varphi(z;\Pb) s_\epsilon(z) \ d\Pb(z) &=   \varphi(z';\Pb) . 
\end{align*}
Therefore this strategy can be described as choosing a clever submodel, so that the pathwise derivative immediately returns the influence function itself, at least in discrete models. (Though again, one can either use the approach of \citet{ichimura2022influence} to generalize, or else check that the pathwise differentiability condition \eqref{eq:pathwise} holds in the general case for the putative influence function that is derived.) \\

In what follows we show an example of using this strategy for the regression function parameter in the discrete case. We refer to \citet{hines2022demystifying} for more examples, including the expected density and average treatment effect. \\

\begin{example}[continues=ex:reg]
Now we compute the influence function for $\E(Y \mid X=x)$ in the discrete case, using  Strategy 1 via the Gateaux derivative. Let $\delta_z$ be the Dirac measure at $Z=z$, and note that for the submodel $\Pb_\epsilon(Z=z)=(1-\epsilon) \Pb(Z=z) + \epsilon \delta_{z'}$ we have
$$ \Pb_\epsilon(Y=y \mid X=x) = \frac{\Pb_\epsilon(Z=z) }{\Pb_\epsilon(X=x)} = \frac{ (1-\epsilon) \Pb(Z=z) + \epsilon \one(z=z') }{ (1-\epsilon)  \Pb(X=x) + \epsilon \one(x=x') } . $$ 
Therefore the Gateaux derivative is
\begin{align*}
\frac{d}{d\epsilon} \psi\{ (1-\epsilon) \Pb(z) &+ \epsilon \delta_{z'} \} \Bigm|_{\epsilon=0} = \frac{d}{d\epsilon} \sum_y y \frac{ (1-\epsilon) \Pb(Z=z) + \epsilon \one(z=z') }{ (1-\epsilon)  \Pb(X=x) + \epsilon \one(x=x') } \Bigm|_{\epsilon=0} \\
&= \sum_y y \frac{ \{ \one(z=z') - \Pb(Z=z) \} \Pb(X=x) - \{ \one(x=x') - \Pb(X=x) \} \Pb(Z=z) }{ \Pb(X=x)^2 } \\
&= \sum_y y \left\{ \frac{ \one(z=z') - \Pb(Z=z) }{ \Pb(X=x) } - \frac{ \one(x=x') - \Pb(X=x) }{ \Pb(X=x) } \Pb(Y=y \mid X=x)  \right\} \\
&= \sum_y y \left\{ \frac{ \one(z=z') - \one(x=x') \Pb(Y=y \mid X=x) }{ \Pb(X=x) } \right\} \\
&= \frac{y' \one(x=x')}{\Pb(X=x)} - \frac{\one(x=x') \E(Y \mid X=x)}{\Pb(X=x)} = \varphi(z';\Pb)
\end{align*}
where the second equality follows from the quotient rule, and the rest by rearranging. This gives the result.
\end{example}

\bigskip

As illustrated above, the Gateaux derivative strategy is more constructive and direct, and only requires simple derivative calculations. \\

\subsubsection{Strategy 2}

Strategy 2 is similar in spirit to Strategy 1, but allows for extra shortcuts and can bypass unnecessary derivative calculations required in the standard Gateaux derivative approach. The main idea consists of the following tricks:
\begin{enumerate}
\item[\textsc{Trick} 1.] Pretend the data are discrete.
\item[\textsc{Trick} 2.]  Treat influence functions as derivatives, allowing use of differentiation rules.
\item[\textsc{Trick} 3.] Use influence function building blocks, e.g., that the influence function of  $\E(Y \mid X=x)$ is $\frac{\one(X=x)}{\Pb(X=x)} \{Y - \E(Y \mid X=x) \}$. 
\end{enumerate}

To help make ideas concrete, we  introduce an operator $\IF: \Psi \rightarrow L_2(\Pb)$ that maps functionals $\psi: \mathcal{P} \rightarrow \R$ to their influence functions $\varphi(z) \in L_2(\Pb)$ in a nonparametric model. Then Trick 2 can for example include 
\begin{enumerate}
\item[\textsc{Trick} 2a.] \emph{(product rule)} $\IF(\psi_1 \psi_2) = \IF(\psi_1) \psi_2 + \psi_1 \IF(\psi_2)$
\item[\textsc{Trick} 2b.] \emph{(chain rule)} $\IF(f(\psi)) = f'(\psi) \IF(\psi)$
\end{enumerate}
and Trick 3 can be written as $\IF(\E(Y \mid X=x)) = \frac{\one(X=x)}{\Pb(X=x)} \{Y - \E(Y \mid X=x) \}$. \\

For comparison, we include derivations of the influence function for the average treatment effect using the  general indirect approach and the Gateaux approach of Strategy 1 in the Appendix. The former requires about two pages of calculations, and the latter about one page. Contrast this with the following calculations, which only comprise \emph{four lines} (in addition to not requiring explicit submodels or complicated derivative calculations). \\

\begin{example}[continues=ex:ate]
Let $\mu(x) = \E(Y \mid X=x,A=1)$, $\pi(x) = \Pb(A=1 \mid X=x)$, and $p(x) = \Pb(X=x)$, and let $\psi=\E\{\E(Y \mid X,A=1)\}$ denote the average treatment effect functional. Then the influence function is given by
\begin{align*}
\IF(\psi) &= \IF\left\{ \sum_{x} \mu(x) p(x) \right\} = \sum_{x} \Big[ \IF\{  \mu(x) \} p(x) +  \mu(x) \IF \{ p(x) \} \Big] \\
&=  \sum_{x} \frac{\one(X=x,A=1)}{p(1,x)} \Big\{ Y - \mu(x) \Big\}  p(x) +  \mu(x) \{ \one(x=X) - p(x) \} \Big] \\
&= \frac{A}{\pi(X)} \Big\{ Y - \mu(X) \Big\} + \mu(X) - \psi 
\end{align*}
where the first equality follows by Trick 1, the second by Trick 2a, the third by Trick 3, and the fourth by rearranging. 
\end{example}

\bigskip

\begin{example}[continues=ex:stochastic]
Let $\psi = \int \int \mu(x,a) dG(a \mid x)  \ d\Pb(x)$ denote the stochastic intervention effect, where $\mu(x,a)=\E(Y \mid X=x,A=a)$ and $\pi(a \mid x) = \Pb(A=a \mid X=x)$ as usual. 
Then 
\begin{align*}
\IF(  \psi ) &= \IF\left\{ \sum_{x,a} \mu(x,a) g(a \mid x) p(x) \right\} \\
&=\sum_{x,a}  \Big[ \IF\{ \mu(x,a) \} g(a \mid x) p(x) + \mu(x,a)  g(a \mid x) \IF\{ p(x)\} \Big] \\
&=\sum_{x,a}  \left[ \frac{\one(A=a,X=x)}{\pi(a \mid x)p(x) } \Big\{ Y - \mu(x,a) \Big\} g(a \mid x) p(x) + \mu(x,a)  g(a \mid x) \Big\{ \one(X=x) - p(x) \Big\}  \right] \\
&= \frac{g(A \mid X)}{\pi(A \mid X)} \Big\{ Y - \mu(X,A) \Big\} + \sum_a \mu(X,a) g( a \mid X) - \psi 
\end{align*}
where the first equality follows by Trick 1, the second  by Trick 2a, the third by Trick 3, and the fourth rearranging.
In general when $A^* \sim dG(a \mid x)$ this influence function would be
$$ \frac{g(A \mid X)}{\pi(A \mid X)} \Big\{ Y - \mu(X,A) \Big\} + \int \mu(X,a) \ dG( a \mid X)  - \psi . $$ 
\end{example}

\bigskip

\begin{example}[continues=ex:iv]
Let 
$ \psi = \frac{\E\{\E(Y \mid X, R=1) - \E(Y \mid X,R=0)\}}{\E\{\E(A \mid X, R=1) - \E(A \mid X,R=0)\}}  \equiv \frac{\E\{\mu(X,1) - \mu(X,0)\}}{\E\{\eta(X,1)-\eta(X,0)\}}$ denote the local average treatment effect with instrument $R$. 
First note that $\psi = \psi_{iv,num}/\psi_{iv,den}$ where $\psi_{iv,num} = \E(Y^{R=1}-Y^{R=0})$ and $\psi_{iv,den}=\E(A^{R=1}-A^{R=0})$, so that 
\begin{align*}
\varphi_{iv,num} \equiv \IF(\psi_{iv,num}) & =  \frac{2R-1}{\varpi(R \mid X)} \Big\{ Y - \mu(X,R) \Big\} + \mu(X,1) - \mu(X,0) - \psi_{iv,num} \\
\varphi_{iv,den} \equiv \IF(\psi_{iv,den}) &=  \frac{2R-1}{\varpi(R \mid X)} \Big\{ A - \eta(X,R) \Big\} + \eta(X,1) - \eta(X,0) - \psi _{iv,den}
\end{align*}
for $\varpi(r \mid x) = \Pb(R=r \mid X=x)$. Therefore
\begin{align*}
\IF(\psi) &= \IF\left( \frac{\psi_{iv,num}}{\psi_{iv,den}} \right) = \frac{\IF(\psi_{iv,num})}{\psi_{iv,den}} - \left( \frac{\psi_{iv,num}}{\psi_{iv,den}} \right)\frac{\IF(\psi_{iv,den})}{\psi_{iv,den}} \\
&= \frac{1}{\psi_{iv,den}} \bigg( \frac{2Z-1}{\varpi(Z \mid X)} \Big\{ Y - \mu(X,Z) \Big\} + \mu(X,1) - \mu(X,0) \\
& \hspace{.5in} -  \psi \left[ \frac{2Z-1}{\varpi(Z \mid X)} \Big\{ A - \eta(X,Z) \Big\} + \eta(X,1) - \eta(X,0) \right] \bigg)
\end{align*}
where the second equality follows by Trick 2a, and the third by Trick 3. 
\end{example} 

\bigskip

\begin{example}[continues=ex:gformula]
Let $\mu_{11}(x_2,x_1) = \E(Y \mid A_2=1,X_2=x_2, A_1=1,X_1=x_1)$, $\pi_t(h_t) = \Pb(A_t=1 \mid H_t=h_t)$ for $H_t = (\overline{X}_t,\overline{A}_{t-1})$, and let
$$ \psi \equiv \E(Y^{11}) = \int \int \E(Y \mid A_2=1,X_2, A_1=1,X_1) \ d\Pb(X_2 \mid A_1=1,X_1) \ d\Pb(X_1)$$
denote the g-formula functional. 
Then using the same logic as in previous examples, we have
\begin{align*}
\IF(\psi) &= \IF\left\{ \sum_{x_1,x_2} \E(Y \mid A_2=1,X_2=x_2, A_1=1,X_1=x_1) p(x_2 \mid a_1=1, x_1) p(x_1) \right\} \\
&= \sum_{x_1,x_2} \bigg[ \IF\Big\{ \E(Y \mid A_2=1,X_2=x_2, A_1=1,X_1=x_1) \Big\} p(x_2 \mid a_1=1, x_1) p(x_1) \\
& \hspace{.5in}  +  \E(Y \mid A_2=1,X_2=x_2, A_1=1,X_1=x_1) \IF\Big\{ p(x_2 \mid a_1=1, x_1) \Big\}  p(x_1) \\
& \hspace{.5in} +   \E(Y \mid A_2=1,X_2=x_2, A_1=1,X_1=x_1) p(x_2 \mid a_1=1, x_1) \IF\Big\{ p(x_1)\Big\} \bigg] \\
&= \sum_{x_1,x_2} \bigg[ \frac{A_2 A_1 \one(X_2=x_2,X_1=x_1)}{\pi_2(h_2) \pi_1(h_1)} \Big\{ Y - \mu_{11}(x_2,x_1) \Big\} \\
&\hspace{.5in} + \mu_{11}(x_2,x_1) \frac{A_1 \one(X_1=x_1)}{\pi_1(h_1)} \Big\{ \one(X_2=x_2) - p(x_2 \mid a_1=1,x_1) \Big\} \\
& \hspace{.5in} + \mu_{11}(x_2,x_1) p(x_2 \mid a_1 = 1, x_1) \Big\{ \one(X_1=x_1) - p(x_1) \Big\} \bigg] \\
&= \frac{A_2 A_1}{\pi_2(H_2) \pi_1(H_1)} \Big\{ Y - \mu_{11}(X_2,X_1) \Big\} +  \frac{A_1 }{\pi_1(H_1)} \Big[  \mu_{11}(X_2,X_1)  - \E\{ \mu(X_2,X_1) \mid A_1=1,X_1\} \Big] \\
& \hspace{.5in} + \E\{ \mu(X_2,X_1) \mid A_1=1,X_1\} - \psi .
\end{align*}
\end{example}

\bigskip

\begin{example}[continues=ex:expdens]
For  $\psi=\E\{p(Z)\}$ the expected density functional we have
\begin{align*}
\IF(\psi) &= \IF\left\{ \sum_x p(x)^2 \right\} = \sum_x 2 p(x) \IF\{ p(x) \} \\
&= \sum_x 2 p(x) \Big\{ \one(X=x) - p(x) \Big\} = 2 \Big\{ p(X) - \psi \Big\}
\end{align*}
where the first equality follows by Trick 1, the second by Trick 2b, and the third by Trick 3.
\end{example}

\bigskip

\section{Methods: Influence Function-Based Estimators}
\label{sec:methods}

We now have a generic minimax lower bound, i.e., benchmark for efficient estimation, in nonparametric models (e.g., Theorem \ref{thm:minimax}). Further we have some simple practical tools for deriving efficient influence functions, which are the key components in these minimax lower bounds.  However at this point  nothing has been said about whether these bounds are actually attainable in any generality with real estimators. This is the goal of the present section. \\

\subsection{Using IFs to Correct Plug-In Estimators}

Recall the von Mises (i.e., distributional Taylor) expansion \eqref{eq:vonmises}, in which the functional $\psi: \mathcal{P} \mapsto \R$ satisfies
\begin{equation}
\psi(\overline{P}) - \psi(P) = \int \varphi(z;\overline{P}) \ d(\overline{P}-P)(z) + R_2(\overline{P},P)
\end{equation}
for distributions $\overline{P}$ and $P$, where $\varphi(z;P)$ is a mean-zero, finite-variance function satisfying $\int \varphi(z;P) \ dP(z) =0$ and $\int \varphi(z;P)^2 \ dP(z) < \infty$, and $R_2(\overline{P},P)$ is a \emph{second-order remainder} term (which means it only depends on \emph{products} or \emph{squares} of differences between $\overline{P}$ and $P$). This expansion suggests that generic \emph{plug-in estimators} of the form $\widehat\psi_{pi} = \psi(\widehat\Pb)$ have a first-order bias, since evaluating the expansion at $(\widehat\Pb,\Pb)$ gives
$$ \psi(\widehat\Pb) - \psi(\Pb) = - \int \varphi(z;\widehat\Pb) \ d\Pb(z) + R_2(\widehat\Pb,\Pb) $$
after noting that $\int \varphi(z;\widehat\Pb) \ d\widehat\Pb(z) = 0$ since the influence curve $\varphi$ has mean zero. The following example illustrates with the average treatment effect parameter. \\

\begin{example}[continues=ex:ate]
A plug-in estimator for the average treatment effect functional $\psi = \E\{ \E(Y \mid X,A=1)\}$ is given by $\widehat\psi_{pi} = \Pn\{ \widehat\mu(X) \}$, for $\widehat\mu(x)$ an estimator of $\mu(x) = \E(Y \mid X=x, A=1)$. Suppose for simplicity that $\widehat\mu$ is estimated on a separate sample independent of the sample on which $\Pn$ operates. Then the bias of this plug-in estimator is given by
$$ \E( \widehat\psi_{pi} - \psi) =  \int \E\{ \widehat\mu(x) - \mu(x) \} \ d\Pb(x) , $$
which is just the integrated bias of the regression estimator $\widehat\mu$ itself. For generic estimators $\widehat\mu$ this integrated bias would be expected to be of the same order as the say pointwise bias itself, and so in large nonparametric models with standard tuning (e.g., via cross-validation) would be larger than $1/\sqrt{n}$. Intuitively, this plug-in estimator (if used without special tuning) essentially makes the problem of parameter estimation as hard as regression estimation, whereas the results of the previous section suggest the former should be easier (e.g., in terms of smaller mean squared errors, of the order $1/\sqrt{n}$, being achievable). 
\end{example}

\bigskip

Crucially, the expansion \eqref{eq:vonmises} also suggests how to correct or de-bias generic plug-in estimators, namely by estimating the bias term $-\int \varphi(z;\widehat\Pb) \ d\Pb(z)$ and subtracting it off. Since this expression is just a mean, a natural estimator is given by the corresponding sample average $-\Pn\{ \varphi(Z;\widehat\Pb)\}$, leading to the bias-corrected estimator
\begin{equation} \label{eq:onestep}
\widehat\psi = \psi(\widehat\Pb) + \Pn\{ \varphi(Z;\widehat\Pb) \}  .
\end{equation}
This estimator is also often called a \emph{one-step estimator}, and can be viewed as a generalization of Newton methods for mimicking maximum likelihood estimators in parametric models \citep{pfanzagl1982contributions, bickel1993efficient}.   \\

\begin{example}[continues=ex:ate]
The one-step estimator for the average treatment effect functional $\psi = \E\{\E(Y \mid X,A=1)\}$ is
$$ \widehat\psi = \Pn \left[\widehat\mu(X) + \frac{A\{Y - \widehat\mu(X)\}}{\widehat\pi(X)} \right] $$
where $\widehat\mu$ and $\widehat\pi$ are estimators of $\mu(x)=\E(Y \mid X=x,A=1)$ and $\pi(x) = \Pb(A=1 \mid X=x)$. 
\end{example}

\bigskip

\begin{example}[continues=ex:stochastic]
Let $\psi = \int \int \mu(x,a) dG(a \mid x)  \ d\Pb(x)$ denote the stochastic intervention effect, where $\mu(x,a)=\E(Y \mid X=x,A=a)$ and $\pi(a \mid x) = \Pb(A=a \mid X=x)$. 
The one-step estimator is given by
$$ \widehat\psi = \Pn \left[ \sum_a \widehat\mu(X,a) g( a \mid X)  +  \frac{g(A \mid X)}{\widehat\pi(A \mid X)} \Big\{ Y - \widehat\mu(X,A) \Big\} \right] . $$
\end{example}

\bigskip

\begin{example}[continues=ex:iv]
Let 
$ \psi = \frac{\E\{\E(Y \mid X, R=1) - \E(Y \mid X,R=0)\}}{\E\{\E(A \mid X, R=1) - \E(A \mid X,R=0)\}}  \equiv \frac{\E\{\mu(X,1) - \mu(X,0)\}}{\E\{\eta(X,1)-\eta(X,0)\}}$ denote the local average treatment effect with instrument $R$. 
The one-step estimator is given by
\begin{align*}
\widehat\psi &= \frac{\Pn\{\widehat\mu(X,1) - \widehat\mu(X,0)\}}{\Pn\{\widehat\eta(X,1)-\widehat\eta(X,0)\}} \left[ 1 - \frac{\Pn\{ \varphi_{den}(Z;\widehat\Pb)  \} }{\Pn\{\widehat\eta(X,1)-\widehat\eta(X,0)\}}  \right]+ \frac{\Pn\{ \varphi_{num}(Z;\widehat\Pb) \} }{\Pn\{\widehat\eta(X,1)-\widehat\eta(X,0)\}} 
\end{align*}
where
\begin{align*}
\varphi_{num}(z;\Pb) &= \frac{2Z-1}{\varpi(Z \mid X)} \Big\{ Y - \mu(X,Z) \Big\} + \mu(X,1) - \mu(X,0) \\
\varphi_{den}(z;\Pb) &=\frac{2Z-1}{\varpi(Z \mid X)} \Big\{ A - \eta(X,Z) \Big\} + \eta(X,1) - \eta(X,0) .
\end{align*}
Alternatively one could use the one-step estimators for the numerator and denominator separately, yielding $\widehat\psi = \Pn\{ \varphi_{num}(Z;\widehat\Pb) \}  / \Pn\{ \varphi_{den}(Z;\widehat\Pb) \} $.
\end{example} 

\bigskip

\begin{example}[continues=ex:gformula]
Let $\mu_{11}(x_2,x_1) = \E(Y \mid A_2=1,X_2=x_2, A_1=1,X_1=x_1)$, $\pi_t(h_t) = \Pb(A_t=1 \mid H_t=h_t)$ for $H_t = (\overline{X}_t,\overline{A}_{t-1})$, and let
$$ \psi \equiv \E(Y^{11}) = \int \int \E(Y \mid A_2=1,X_2, A_1=1,X_1) \ d\Pb(X_2 \mid A_1=1,X_1) \ d\Pb(X_1)$$
denote the g-formula functional. The one-step estimator is given by
\begin{align*}
\widehat\psi &= \Pn \bigg( \frac{A_2 A_1}{\widehat\pi_2(H_2) \widehat\pi_1(H_1)} \Big\{ Y - \widehat\mu_{11}(X_2,X_1) \Big\} +  \frac{A_1 }{\widehat\pi_1(H_1)} \Big[  \widehat\mu_{11}(X_2,X_1)  - \widehat\E\{ \widehat\mu(X_2,X_1) \mid A_1=1,X_1\} \Big] \\
& \hspace{.8in} + \widehat\E\{ \widehat\mu(X_2,X_1) \mid A_1=1,X_1\}  \bigg) .
\end{align*}
\end{example}

\bigskip

\begin{example}[continues=ex:expdens]
The one-step estimator for the expected density functional $\psi = \E\{p(Z)\}$ is
$$ \widehat\psi =  2 \Pn\{ \widehat{p}(Z) \}  - \int \widehat{p}(z)^2 \ dz . $$
\end{example}

\bigskip

The one-step estimator \eqref{eq:onestep} can be analyzed in some generality, as it is a simple average of an estimated function. By definition we have the important decomposition
\begin{align}
\widehat\psi - \psi &=  \psi(\widehat\Pb) + \Pn\{ \varphi(Z;\widehat\Pb) \} - \psi(\Pb) \nonumber \\
&=  (\Pn-\Pb) \{ \varphi(Z;\widehat\Pb) \} + R_2(\widehat\Pb,\Pb) \nonumber \\
&=  (\Pn-\Pb) \{  \varphi(Z;\Pb)\} + (\Pn-\Pb) \{ \varphi(Z;\widehat\Pb) - \varphi(Z;\Pb)\} + R_2(\widehat\Pb,\Pb) \nonumber \\
&\equiv S^* + T_1 + T_2 \label{eq:decomp}
\end{align}
where the first line follows by definition of the one-step estimator $\widehat\psi$, the second by the expansion \eqref{eq:vonmises}, and the third after adding and subtracting $(\Pn-\Pb) \{  \varphi(Z;\Pb)\} $. The first term 
$$ S^* = (\Pn-\Pb) \{  \varphi(Z;\Pb)\}  $$
is a simple sample average of a fixed function, and so by the central limit theorem, for example, it behaves as a normally distributed random variable with variance $\var(\varphi)/n$, up to error $o_\Pb(1/\sqrt{n})$. The second term 
$$ T_1 = (\Pn-\Pb) \{ \varphi(Z;\widehat\Pb) - \varphi(Z;\Pb)\} $$ 
is often called an empirical process term, and is typically of smallest order since it is a sample average of a term with shrinking variance (as long as $\varphi(z;\widehat\Pb)$ converges to $\varphi(z;\Pb)$ in a sense to be made formal shortly). The third term
$$ T_2 = R_2(\widehat\Pb,\Pb) = \psi(\widehat\Pb) - \psi(\Pb) + \int \varphi(z;\widehat\Pb) \ d\Pb(z) $$
is the really crucial one. For non-bias-corrected plug-in estimators, this term will generally dominate, but for one-step estimators it will generally involve second-order products of errors, which can be negligible under nonparametric conditions (such as sparsity or smoothness). \\

\begin{remark}[Alternatives to One-Step Correction]
Although the above one-step estimator is intuitive and relatively straightforward to analyze, it is not the only way to construct efficient estimators of pathwise differentiable functionals in nonparametric models. For example, one alternative (which sometimes reduces to one-step estimation) is to solve an \emph{estimating equation} of the form 
$$ \Pn\{ \varphi(z;\widehat\Pb,\psi) \} = 0 $$
in $\psi$, where we write the influence curve as $\varphi(z;\widehat\Pb)=\varphi(z;\widehat\Pb,\psi)$ to stress that in general it depends on the parameter of interest $\psi$. Of course if the influence curve is linear in the parameter, i.e., taking the form $\varphi(z;\Pb) = \phi(z;\Pb) - \psi(\Pb)$, then the estimating equation approach is equivalent to the one-step correction; however in the general nonlinear case these could lead to distinct estimators. Another alternative to one-step correction is to use targeted maximum likelihood estimation (TMLE) \citep{van2006targeted, van2011targeted}. TMLE does not correct bias on the parameter scale by adding an estimate of bias to the plug-in estimator; instead, it aims to correct bias on the distributional scale, by constructing a fluctuated estimate $\widehat\Pb^*$ for which $\Pn\{ \varphi(Z;\widehat\Pb^*) \} \approx 0$, so that 
$$ \psi(\widehat\Pb^*) \approx \psi(\widehat\Pb^*) + \Pn\{ \varphi(Z;\widehat\Pb^*) \} , $$
i.e., a plug-in estimator based on the fluctuated distribution $\widehat\Pb^*$ solves the efficient influence curve estimating equation and behaves like a one-step estimator asymptotically. Despite its asymptotic equivalence to the one-step estimator, an argument for using TMLE is that it could give better finite-sample properties, for example if $\psi(\Pb)$ and $\psi(\widehat\Pb^*)$ are bounded, e.g., in $[0,1]$. A simple one-step estimator can potentially lie outside such bounds on the parameter space depending on the behavior of the correction term $\Pn\{ \varphi(z;\widehat\Pb)\}$. 
\end{remark}

\bigskip

Based on the decomposition \eqref{eq:decomp}, the task of analyzing the one-step estimator $\widehat\psi$ (e.g., deriving its rate of convergence and limiting distribution, and determining if and when it attains the nonparametric effiency bound of Theorem \ref{thm:minimax}) boils down to understanding the behavior of the empirical process term $T_1$ and bias term $T_2$. \\

In particular, when the $T_1$ and $T_2$ terms are of the order $o_\Pb(1/\sqrt{n})$, then the sample average term $S^*$ dominates the decomposition, and so
$$ \sqrt{n}( \widehat\psi - \psi )= \sqrt{n}S^* + o_\Pb(1) \ \indist \ N\Big(0,\var\{\varphi(Z;\Pb)\} \Big) $$
by the central limit theorem and Slutsky's theorem. Such a conclusion would yield several crucial insights, including:
\begin{enumerate}
\item $\widehat\psi$ is root-n consistent, 
\item $\widehat\psi$ is asymptotically normal, with asymptotically valid 95\% confidence intervals for $\psi$ given by the closed-form expression $\widehat\psi \pm 1.96 \sqrt{\widehat\var\{\varphi(Z;\widehat\Pb)\}/n}$, 
\item $\widehat\psi$ is efficient in the local asymptotic minimax sense of Theorem \ref{thm:minimax}.
\end{enumerate}

\bigskip

The next two subsections detail conditions under which the terms $T_1$ and $T_2$ can be negligible relative to $S^*$, even in large nonparametric models where one only assumes some smoothness or sparsity, for example. \\

\subsection{Empirical Process Term $T_1$}

There are two main approaches for arguing that the empirical process term 
$$ T_1 = (\Pn-\Pb) \{ \varphi(Z;\widehat\Pb) - \varphi(Z;\Pb)\} $$ 
is of the order $o_\Pb(1/\sqrt{n})$: one is based on assuming  the function class $\{\varphi(z; P) : P \in \mathcal{P}\}$ and corresponding estimators are not too complex (e.g., Donsker), and the other is to use sample splitting. Both approaches can be viewed as ways to avoid a certain kind of overfitting, as will be discussed in detail shortly. \\

 Regardless of which of these two approaches is used, at minimum it is generally also required that $\varphi(Z; \widehat\Pb)$ be converging to $\varphi(Z;\Pb)$ in $L_2(\Pb)$ norm, i.e., that
\begin{equation} \label{eq:t1consistency}
\| \varphi(;\widehat\Pb) - \varphi(;\Pb) \|^2 \equiv \int \Big\{ \varphi(z;\widehat\Pb) - \varphi(z;\Pb) \Big\}^2 \ d\Pb(z) = o_\Pb(1) . 
\end{equation}
Some intuition for this latter requirement is that if $\varphi(Z; \widehat\Pb)$ is converging to $\varphi(Z;\Pb)$, then $T_1$ is a sample average of a quantity tending to zero, and so would not only be bounded after scaling by $\sqrt{n}$ (i.e., of order $O_\Pb(1/\sqrt{n})$), but tending to zero in probability (i.e., of order $o_\Pb(1/\sqrt{n})$). In general this would be satisfied if $\widehat\Pb$ converges to $\Pb$ (or for relevant estimated components appearing in $\varphi$) and if $\varphi(;P)$ is smooth in $P$.  The next example illustrates with the average treatment effect functional. \\

\begin{example}[continues=ex:ate]
For the  average treatment effect functional $\psi = \E\{\E(Y \mid X,A=1)\}$ the one-step estimator is given by
$$ \widehat\psi = \Pn \left[\widehat\mu(X) + \frac{A\{Y - \widehat\mu(X)\}}{\widehat\pi(X)} \right] $$
and we have 
\begin{align*}
T_1 &= (\Pn-\Pb) \left[\widehat\mu(X) + \frac{A\{Y - \widehat\mu(X)\}}{\widehat\pi(X)} -  \mu(X) - \frac{A\{Y - \mu(X)\}}{\pi(X)}   \right] \equiv (\Pn-\Pb) \{ \widehat{f}(Z) - f(Z)\}
\end{align*}
since $(\Pn-\Pb)(\widehat\psi-\psi) = (\widehat\psi-\psi)(\Pn-\Pb)(1)= 0$. Now note that
\begin{align*}
\widehat{f}-f &= \left( 1- \frac{A}{\pi} \right) (\widehat\mu-\mu)+  \frac{A(Y - \widehat\mu)}{\widehat\pi \pi} ( \pi - \widehat\pi  ) 
\end{align*}
and so one set of simple sufficient conditions for \eqref{eq:t1consistency} to hold for $(\widehat{f}-f)$ is that
\begin{enumerate}
\item $\pi(x) \geq \epsilon$ and $\widehat\pi(x) \geq \epsilon$ with probability one, for some $\epsilon>0$,
\item $|Y - \widehat\mu| \leq C$  with probability one, for some $C < \infty$, and
\item $\| \widehat\mu - \mu \| = o_\Pb(1)$ and $\| \widehat\pi - \pi \| = o_\Pb(1)$,
\end{enumerate}
since under these conditions it follows that
\begin{align*}
\|\widehat{f}-f \| & \leq \left( 1 + \frac{1}{\epsilon} \right) \| \widehat\mu-\mu\| + \left(  \frac{C}{\epsilon^2} \right) \| \pi - \widehat\pi  \| . 
\end{align*}
Note boundedness of $|Y-\widehat\mu|$ could be relaxed to bounded moment conditions, as usual (e.g., using Holder's inequality).
\end{example}

\bigskip

\begin{remark}
For so-called doubly robust influence functions,  it can in some cases be enough to argue $\overline{T}_1=o_\Pb(1/\sqrt{n})$ for $\overline{T}_1=(\Pn-\Pb) \{ \varphi(Z;\widehat\Pb) - \varphi(Z;\overline\Pb)\} $, where only some components of $\overline\Pb$ equal $\Pb$, while others can merely be set to whatever corresponding estimators converge to. For the average treatment effect functional, for example, one may only have consistency of $\widehat\pi$ but not $\widehat\mu$, in which case one could define $\varphi(Z;\overline\Pb)= \overline\mu(X) + \frac{A\{Y - \overline\mu(X)\}}{\pi(X)} $, where $\overline\mu \neq \mu$ is defined as the misspecified limit of $\widehat\mu$. However, in this case the influence function $\varphi(Z;\overline\Pb)$ would not be the efficient one, and in general the $T_2$ term would be too large to be of order $o_\Pb(1/\sqrt{n})$, and so would contribute to the limiting distribution. If $\widehat\pi$ and $\widehat\mu$ were estimated with parametric models, the contribution from the $T_2$ term could behave like a sample average asymptotically, but for nonparametric estimators this would not hold in general, and so there the rate of convergence would in general degrade from $1/\sqrt{n}$ to something slower, depending on the rate at which $\pi$ was estimated. 
\end{remark}

\bigskip

Now we briefly describe a first approach for analyzing the empirical process term $T_1$ in nonparametric models, which is based on assuming the function class $\{\varphi(z; P) : P \in \mathcal{P}\}$ and corresponding estimators are not too complex (e.g., Donsker). We only briefly describe this approach for two primary reasons: (i) as the more classical approach, there are already widely available references \citep{andrews1994empirical, van1996weak, van2000asymptotic, van2002semiparametric, kosorok2008introduction, kennedy2016semiparametric}, (ii) the second way to control the term $T_1$, using sample splitting, is much simpler and requires weaker assumptions. Nonetheless we give some intuition for the main idea here. \\

Intuitively, when not using sample splitting (i.e., when $\widehat\Pb$ is estimated on the same sample on which $\Pn$ operates),  the bias correction in the one-step estimator \eqref{eq:onestep} is using the same data twice, for two different tasks, in a kind of ``double-dipping''. One task is to construct relevant nuisance components in $\widehat\Pb$, and the other to estimate the bias term $\Pn\{ \varphi(Z;\widehat\Pb)\}$. This double-dipping introduces a threat of overfitting. A nice illustration of this overfitting can be found in Figure 2 of \citet{chernozhukov2018double}, and the surrounding discussion. One natural way to avoid overfitting in general is to not fit overly complex models; this is precisely what a Donsker-type assumption on the complexity of $\Pb$ and $\widehat\Pb$ achieves. In particular, Donsker classes include smooth parametric models, but also  bounded monotone functions, smooth functions with bounded partial derivatives, Sobolev classes, functions with bounded sectional variation, etc.\ (\citet{van2000asymptotic} has a nice review in Chapter 19). \\

However, Donsker assumptions can still be restrictive, as noted for example by \citet{robins2008higher} (Remark 2.8), \citet{zheng2010asymptotic}, and \citet{chernozhukov2018double}, and may require avoiding commonly used methods such as lasso or random forests. \citet{chernozhukov2018double} points out how high-dimensional models can fail to be Donsker, and more generally have large entropy unless one imposes potentially overly strict sparsity assumptions. \\

Luckily, there is a straightforward alternative to employing Donsker-type conditions, which is simpler to analyze despite requiring weaker assumptions: sample-splitting (and its swapped analog, now commonly referred to as cross-fitting). Sample-splitting allows one to completely avoid complexity restrictions (only requiring consistency \eqref{eq:t1consistency}), but also greatly simplifies proofs. The latter advantage seems to have driven its initial use in functional estimation problems, as in, e.g., \citet{hasminskii1978nonparametric},  \citet{pfanzagl1982contributions}, \citet{schick1986asymptotically}, \citet{bickel1988estimating}, etc. It is important to note that although using sample-splitting and cross-fitting for handling terms like $T_1$ has become popular recently, it does have a long history, going back nearly half a century, at least. \\

Using sample splitting has some straightforward and simple intuition in this context:  to avoid the overfitting that can come with the aforementioned ``double-dipping'' (i.e., using the data twice, once to estimate $\widehat\Pb$ and again to estimate the mean $\Pn\{ \varphi(Z;\widehat\Pb)\}$), just formally separate these two estimation tasks, performing them on different independent samples. \\

More specifically, cross-fitting works as follows. First randomly split observations $Z^n = (Z_1,...,Z_n)$ into $K$ disjoint  folds. This can be formalized notationally via $n$ realizations of a random variable $F \in \{1,...,K\}$, drawn independently of the data $Z^n$, where $F_i=k$ means subject $i$ is assigned to fold $k$. Then we can let $\widehat\Pb_{-k}$ denote an estimator of $\Pb$ (or its relevant components appearing in the influence curve $\varphi$) that only uses observations $F_i \neq k$, i.e., excludes fold $k$. Note that there will be $K$ different such estimators, since there are $K$ folds.  Then, rather than constructing the estimator in \eqref{eq:onestep}, where $\widehat\Pb$ and $\Pn$ are built from and operator on the same sample, instead one constructs the estimator
\begin{equation} \label{eq:cfonestep}
\widehat\psi = \sum_{k=1}^K \left( \frac{N_k}{n} \right) \widehat\psi_k 
\end{equation}
where $N_k = \sum_i \one(F_i=k)$ is the number of observations in the $k$th fold and 
\begin{equation}
\widehat\psi_k =\psi(\widehat\Pb_{-k}) + \Pn^k\Big\{ \varphi(Z;\widehat\Pb_{-k}) \Big\}
\end{equation}
is the usual one-step estimator in the $k$th fold, with $\Pn^k f(Z) = N_k^{-1} \sum_{F_i =k} f(Z_i) $ the empirical measure over the $k$th fold. Then for each $\widehat\psi_k$ the decomposition \eqref{eq:decomp} becomes
\begin{align}
\widehat\psi_k - \psi &= \psi(\widehat\Pb_{-k}) + \Pn^k\Big\{ \varphi(Z;\widehat\Pb_{-k}) \Big\} - \psi(\Pb)  \nonumber \\
&=  (\Pn^k-\Pb) \{ \varphi(Z;\widehat\Pb_{-k}) \} + R_2(\widehat\Pb_{-k},\Pb) \nonumber \\
&=  (\Pn^k-\Pb) \{  \varphi(Z;\Pb)\} + (\Pn^k-\Pb) \{ \varphi(Z;\widehat\Pb_{-k}) - \varphi(Z;\Pb)\} + R_2(\widehat\Pb_{-k},\Pb) \nonumber \\
&\equiv S_k^* + T_{1k} + T_{2k} \label{eq:decomp2}
\end{align}
by the exact same logic as before, and similarly
\begin{align*}
\widehat\psi - \psi &= S^* +  \sum_{k=1}^K  \left( \frac{N_k}{n} \right) \Big( T_{1k} + T_{2k} \Big) \equiv S^* + T_1 + T_2
\end{align*}
which follows since $ \sum_{k=1}^K  \left( \frac{N_k}{n} \right) S_k^* = (\Pn-\Pb) \{  \varphi(Z;\Pb)\} = S^*$. If the number of folds $K$ is finite, then the order of the terms $\sum_{k=1}^K  (N_k/n) ( T_{1k} + T_{2k} )$ is the same as $\max_k (T_{1k} + T_{2k})$, and one can just focus on $T_{1k}$ and $T_{2k}$ separately.  Note the number of folds $K$ would not be finite if it scaled with $n$, as in leave-one-out cross-validation, which requires a different analysis. \\

\begin{remark}
Typically one uses equally sized folds so that $N_k=n/K$, at least approximately, in which case $\widehat\psi$ is just an average of the fold-specific estimators $\widehat\psi_k$, and similarly $T_1$ and $T_2$ are averages of the $T_{1k}$ and $T_{2k}$ terms, respectively. 
\end{remark}

\bigskip

Here we assume a fixed number of folds $K$, and so can analyze the terms $T_{1k}$ and $T_{2k}$ on their own (the former here, and the latter in the next subsection). By virtue of the sample splitting, somewhat remarkably, a simple bias-variance analysis combined with Chebyshev's inequality is enough to control the $T_{1k}$ terms. This is illustrated in the following lemma from \citet{kennedy2020sharp}, though the same ideas are found in aforementioned earlier work using sample splitting as well. \\

\begin{lemma}[\citet{kennedy2020sharp}] \label{splitlem}
Let $\widehat{f}(z)$ be a function estimated from a sample $Z^N=(Z_{n+1},\ldots,Z_N)$, and let $\Pn$ denote the empirical measure over $(Z_1,\ldots,Z_n)$, which is independent of $Z^N$. Then 
$$ (\Pn-\Pb) (\widehat{f}-f) = O_\Pb\left( \frac{ \| \widehat{f}-f \| }{\sqrt{n}}  \right) . $$ 
\end{lemma}

\begin{proof}
First note that, conditional on $Z^N$, the term in question has mean zero since
$$ \E\Big\{ \Pn(\widehat{f}-f) \Bigm| Z^N \Big\}  = \E(\widehat{f}-f \mid Z^N) = \Pb(\widehat{f}-f) . $$
The conditional variance is
\begin{align*}
\var\Big\{ (\Pn-\Pb) (\widehat{f}-f) \Bigm| Z^N \Big\} &=  \var\Big\{ \Pn(\widehat{f}-f) \Bigm| Z^N \Big\} = \frac{1}{n} \var(\widehat{f}-f \mid  Z^N ) \leq \|\widehat{f}-f\|^2 /n .
\end{align*}
Therefore by iterated expectation and Chebyshev's inequality we have
\begin{align*}
\Pb\left\{ \frac{ | (\Pn-\Pb)(\widehat{f}-f) | }{ \| \widehat{f}-f \| / \sqrt{n} } \geq t \right\} &= \E\left[ \Pb\left\{ \frac{ | (\Pn-\Pb)(\widehat{f}-f) | }{ \| \widehat{f}-f \| / \sqrt{n} } \geq t \Bigm| Z^N \right\} \right] \leq \frac{1}{t^2} .
\end{align*}
Thus for any $\epsilon>0$ we can pick $t=1/\sqrt{\epsilon}$ so that the probability above is no more than $\epsilon$, which yields the result.
\end{proof}

\bigskip

Thus the above lemma shows how, as long as $\varphi(z;\widehat\Pb)$ is consistent for $\varphi(z;\Pb)$ in $L_2(\Pb)$ norm, and there are finitely many folds $K$, then sample-splitting/cross-fitting ensures that $T_1 = \sum_k (N_k/n) T_{1k} = o_\Pb(1/\sqrt{n})$, which is asymptotically negligible relative to the sample average term $S^*$. Importantly, no complexity-restricting conditions using Donsker classes or entropy bounds are required, which means arbitarily flexible methods can be accomodated (e.g., lasso, random forests), as long as they are consistent. Due to its simplicity, the sample-splitting-based approach is arguably also more transparent, as it only requires reasoning about means and variances, rather than Donsker classes and empirical processes. \\

We summarize the above results in the following proposition. \\

\begin{proposition} \label{prop:t1}
Let $\widehat\psi$ denote the cross-fit estimator in \eqref{eq:cfonestep}. Assume $K \leq C < \infty$ is finite, and that $\| \varphi(z; \widehat\Pb_{-k}) - \varphi(z; \Pb) \| = o_\Pb(1)$ for each $k$.  
Then
$$ \widehat\psi - \psi = (\Pn-\Pb)\{ \varphi(Z;\Pb)\} + T_2 + o_\Pb(1/\sqrt{n}) $$
for $T_2 =  \sum_{k=1}^K  \left( \frac{N_k}{n} \right) R_2(\widehat\Pb_{-k},\Pb)$ and $R_2(\overline{P},P) = \psi(\overline{P}) - \psi(P) + \int \varphi(z;\overline{P}) \ dP(z)$. 
\end{proposition}

\bigskip

\subsection{Remainder Bias Term $T_2$}

We are now a step closer to obtaining estimators that are: (i) root-n consistent, (ii) asymptotically normal, and (iii) efficient in the local asymptotic minimax sense of Theorem \ref{thm:minimax}. The last step is to analyze the bias term $T_2$, which typically needs to be studied on a case-by-case basis. This $T_2$ term is what makes bias-corrected estimators like \eqref{eq:onestep} and \eqref{eq:cfonestep} special, e.g., allows them to be asymptotically unaffected by nuisance estimation. For example, simple plug-in estimators would have similar decompositions as in Proposition \ref{prop:t1}, with similarly small $T_1$ terms, but the analog of the $T_2$ term would in general never be $o_\Pb(1/\sqrt{n})$. In contrast, for influence function-based estimators given above,  the $T_2$ term just equals
$$ T_2 = R_2(\widehat\Pb,\Pb) $$
for $R_2(\overline{P},P) = \psi(\overline{P}) - \psi(P) + \int \varphi(z;\overline{P}) \ dP(z)$
 the remainder of the distributional Taylor expansion \eqref{eq:vonmises} (or $T_2 =  \sum_{k=1}^K  ( \frac{N_k}{n} ) R_2(\widehat\Pb_{-k},\Pb)$ in the cross-fitting case, which is basically the same, as we assume $K$ is finite throughout). Therefore the bias term $T_2$ is essentially a byproduct of deriving and verifying the expansion \eqref{eq:vonmises}, and in such cases will involve second-order \emph{products} of differences between $\widehat\Pb$ and $\Pb$. Thus if each such error is of the order $n^{-1/4}$, the product of will be of the order $1/\sqrt{n}$. In what follows we illustrate with several examples. \\

\begin{example}[continues=ex:ate]
For the average treatment effect or missing outcome functional 
$$ \psi(P) = \E_P \{ \E_P(Y \mid X, A=1)\} $$
the remainder in \eqref{eq:vonmises} is given by
$$ R_2(\overline{P},P) = \int \left\{ \frac{1}{\overline\pi(x)} - \frac{1}{\pi(x)} \right\} \Big\{ \mu(x) - \overline\mu(x) \Big\} \pi(x) \ dP(x) $$
where $\pi(x) = P(A=1 \mid X=x)$ and $\overline\pi(x) = \overline{P}(A=1 \mid X=x)$, and similarly for $\mu(x) = \E_P(Y \mid X=x,A=1)$. Therefore if $\widehat\pi(x) \geq \epsilon$ with probability one we have
$$ | R_2(\widehat\Pb,\Pb) | \leq  \left( \frac{1}{\epsilon} \right) \int| \pi(x) - \widehat\pi(x) | | \mu(x) - \widehat\mu(x) | \ d\Pb(x) \leq \left( \frac{1}{\epsilon} \right) \| \widehat\pi - \pi \| \| \widehat\mu - \mu \| $$
by Cauchy-Schwarz. Thus if $\|\widehat\pi - \pi \| = o_\Pb(n^{-1/4})$ and $\|\widehat\mu - \mu \| = o_\Pb(n^{-1/4})$, for example, then $T_2 = o_\Pb(1/\sqrt{n})$, as desired (though note that $\pi$ and $\mu$ do not have to be estimated at the same rates for $T_2=o_\Pb(1/\sqrt{n})$ to hold: any such combination whose product is $o_\Pb(1/\sqrt{n})$ would suffice). Conditions under which $L_2(\Pb)$ errors satisfy, e.g., $\|\widehat\pi - \pi \| = o_\Pb(n^{-1/4})$ are available for many popular estimators.  For example, if $\pi$ is $s$-smooth (i.e., contained in a \Holder{} class with index $s$, so that all partial derivatives up to order $s-1$ are bounded and the highest order are continuous), and $\widehat\pi$ is a minimax optimal estimator (e.g., using local polynomials) then with appropriate tuning 
$$ \| \widehat\pi - \pi \| = O_\Pb\left(n^{-\frac{1}{2+d/s}} \right) $$
\citep{gyorfi2002distribution, tsybakov2009introduction}, and so for example the rate would be $o_\Pb(n^{-1/4})$ if $d/s<2$, i.e., the smoothness was more than half the dimension. Similarly, if $\pi$ is $s$-sparse and estimated say via lasso at a rate like
$$ \| \widehat\pi - \pi \| = O_\Pb\left( \sqrt{ \frac{s \log d}{n} } \right) $$
\citep{farrell2015robust, chernozhukov2018double, bradic2019sparsity}, then  the rate would be $o_\Pb(n^{-1/4})$ if $s =o(\sqrt{n}/\log d)$, i.e., the sparsity $s$ scales slower than $\sqrt{n}$ up to log factors. Similar results can be obtained for random forests, neural networks, etc., under appropriate conditions \citep{farrell2021deep}.
\end{example}

\bigskip

\begin{example}[continues=ex:cov] For the expected conditional covariance functional
$$ \psi(P) = \E_P\{ \cov_P(A,Y \mid X) \}$$
the remainder in \eqref{eq:vonmises} is given by
$$ R_2(\overline{P},P) = \int \Big\{ \overline\pi(x) - \pi(x) \Big\} \Big\{ \overline\mu(x) - \mu(x) \Big\}   \ dP(x) $$
where $\pi(x) = \E_P(A \mid X=x)$ and  $\mu(x) = \E_P(Y \mid X=x)$. Therefore 
$$ | R_2(\widehat\Pb,\Pb) |  \leq  \| \widehat\pi - \pi \| \| \widehat\mu - \mu \| $$
by Cauchy-Schwarz, and so this functional has the same kind of double robustness properties as the average treatment effect function in Example \ref{ex:ate}.
\end{example}

\bigskip

\begin{example}[continues=ex:expdens]
For the expected density functional
$$ \psi(P) = \E_P\{ p(Z) \} = \int p(z)^2 \ dz $$
the remainder in \eqref{eq:vonmises} is given by
$R_2(\overline{P},P) = - \int \{ \overline{p}(z) - p(z) \}^2   \ dz$, so that
$$ | R_2(\widehat\Pb,\Pb) |  = \| \widehat{p} - p \|^2 $$
where in a slight abuse of notation  $\| \cdot \|$ above denotes the $L_2(\nu)$ norm for $\nu$ the uniform measure. Thus for the expected density, the standard one-step estimator is not doubly robust, but still has nuisance errors that consist of a second-order product; therefore this estimator will still be root-n consistent, asymptotically normal, and efficient as long as the density is estimated at faster than $n^{-1/4}$ rates. 
\end{example}

\bigskip

The above examples illustrate the kinds of arguments one can use to show the bias term $T_2$ is of order $o_\Pb(1/\sqrt{n})$. For some functionals, the remainder term in the expansion \eqref{eq:vonmises} can take multiple possible forms, some of which may be more or less useful depending on context (e.g., consider the expected conditional covariance in Example \ref{ex:cov} and note that $\mu = \pi \mu_1 + (1-\pi) \mu_0$ for $\mu_a(x) = \E(Y \mid X=x, A=a)$), and in some cases the remainder terms can be quite complicated to derive (e.g., for longitudinal causal effects like Example \ref{ex:gformula} when there are more than two timepoints). \\

We summarize the above results in the following proposition. \\

\begin{proposition} \label{prop:t2}
Let $\widehat\psi$ denote the cross-fit estimator in \eqref{eq:cfonestep}. Assume $K \leq C < \infty$ is finite, $\| \varphi(z; \widehat\Pb_{-k}) - \varphi(z; \Pb) \| = o_\Pb(1)$ for each $k$, and that $T_2 \equiv  \sum_{k=1}^K  \left( \frac{N_k}{n} \right) R_2(\widehat\Pb_{-k},\Pb) = o_\Pb(1/\sqrt{n})$, where $R_2(\overline{P},P) = \psi(\overline{P}) - \psi(P) + \int \varphi(z;\overline{P}) \ dP(z)$. 
Then
$$ \widehat\psi - \psi = (\Pn-\Pb)\{ \varphi(Z;\Pb)\} + o_\Pb(1/\sqrt{n}) $$
and so $\widehat\psi$ is root-n consistent, asymptotically normal, and minimax optimal in the local asymptotic sense of Theorem \ref{thm:minimax} if $\varphi$ is the efficient influence function. 
\end{proposition}

\section{Some Extensions \& Open Problems}

Although the theory sketched in previous sections has by now a relatively long history and is quite well-developed, there are many extensions and open problems still remaining, which will be important to further develop in coming years. Here we briefly detail some recent examples.  \\

\subsection{New Functionals}
 
First, every functional has a somewhat unique remainder term $R_2$ in the expansion \eqref{eq:vonmises}, leading to corresponding unique bias term $T_2$ in the decomposition \eqref{eq:decomp} (though see \citet{rotnitzky2019characterization} for some unifying properties). Thus, as new parameters are developed in causal inference and other fields, for example for newly defined causal effects (e.g., \citet{diaz2012population, young2014identification, haneuse2013estimation, kennedy2019nonparametric} for stochastic intervention effects), or with new data structures (e.g.,  \citet{tchetgen2012causal, van2014causal, ogburn2017causal} for network data), or under new identifying assumptions (e.g., \citet{tchetgen2020introduction}), it will be crucial to understand the operating characteristics of these new quantities. This includes understanding terms like $R_2$ in \eqref{eq:vonmises} and $T_2$ in \eqref{eq:decomp}, even if this uses some already well-developed theoretical tools or arguments. One of the fascinating aspects of causal inference is that there are many different ways to characterize causal effects, each with its own nuances and subtleties. \\

\subsection{High Complexity Regimes}

In this review we focused on discussing conditions under which estimators are root-n consistent and asymptotically normal. However, when underlying nuisance functions are not smooth or sparse enough relative to the dimension, root-n consistency will be impossible to attain. For example, if a regression function is $s$-smooth with $s=5$, but the dimension of the covariates is $50$, then the minimax rate is $n^{-1/12}$, much slower than the $n^{-1/4}$ rates discussed in the previous section. This opens up several questions not addressed in this review: When  influence function-based estimators like \eqref{eq:onestep} and \eqref{eq:cfonestep} are not root-n consistent, can any other estimator be? When are root-n rates achievable? What is the best possible rate that can be achieved when root-n rates are \emph{not} achievable? Important progress along these lines has been made for complex causal effect-style functionals in the last decade by \citet{robins2008higher, robins2009quadratic, robins2017minimax}, but many open problems remain, including: (i) minimax rates for other functionals beyond the average treatment effect and expected conditional covariance, and (ii) minimax rates for models with nuisances in non-\Holder{} function spaces \citep{bradic2019minimax}. 

\subsection{Non-Pathwise Differentiable Functionals}

Throughout this review we have considered functionals that satisfy the von Mises / distributional Taylor expansion \eqref{eq:vonmises}. As explained and illustrated throughout, many important parameters in causal inference and other fields \emph{do} satisfy this expansion; however, there are also many parameters that do not. These parameters are often referred to as non-pathwise differentiable. Non-pathwise differentiability arises in at least two prominent settings: (i) non-smooth finite-dimensional parameters, and (ii) infinite-dimensional parameters. By non-smooth finite-dimensional parameters, we mean parameters that resemble those studied in this review (e.g., take values in $\R$ and, e.g.,  can be represented as expectations over nuisance functions), but involve non-differentiable functions of nuisance quantities, such as indicators, maximums, absolute values, etc. These arise often, e.g., in the optimal treatment regime literature \citep{murphy2003optimal, hirano2012impossibility, laber2014dynamic, luedtke2016statistical}. For example, under standard no unmeasured confounding and other assumptions, the value of the mean-optimal treatment regime is given by
$$ \E \Big[ \mu_1(X) \one\{ \mu_1(X) \geq \mu_0(X) \} + \mu_0(X) \one\{ \mu_1(X) < \mu_0(X) \} \Big] $$
where $\mu_a(x) = \E(Y \mid X=x, A=a)$. Since the nuisance functions $\mu_a$ appear inside non-smooth indicator functions, many arguments from the previous sections cannot be applied directly. Interestingly, though, it can be shown that the influence function for the value still exists under a margin condition \citep{van2013targeted, luedtke2016statistical}, but in general this is not the case. \\

A second place where pathwise differentiability fails is for parameters that are not simple expectations, but are instead given by infinite-dimensional curves or functions, like regression and density functions. These arise in many important and common settings in causal inference, for example: effects of continuous treatments \citep{rubin2006extending, diaz2013targeted, kennedy2017nonparametric}, counterfactual density estimation \citep{robins2001inference, kim202xcausal, kennedy202xdensity, westling2020unified},  heterogeneous effect estimation 
\citep{nie2017quasi, semenova2017estimation, foster2019orthogonal, kennedy202xoptimal, kennedy202xminimax}, etc. Interestingly, although these quantities take values in  infinite-dimensional function spaces like densities or regressions, they also involve structured combinations of nuisance quantities, and so are really regression/functional hybrids; it turns out that tools and concepts from standard functional estimation  can therefore be adapted to these settings. For example, consider the conditional average treatment effect parameter
$$ \tau(x) = \E(Y \mid X=x, A=1) - \E(Y \mid X=x, A=0); $$
the function $\tau(x)$ can be very smooth or sparse, even when the individual regression functions it differences are not. Exploiting such smoothness or sparsity when it exists requires adapting ideas described in this review to the infinite-dimensional context. We refer to  \citet{nie2017quasi, semenova2017estimation, foster2019orthogonal, kennedy202xoptimal, kennedy202xminimax} for some examples of such adaptation. \\

\section*{Acknowledgements}

EK gratefully acknowledges support from NSF Grants DMS1810979 and CAREER Award
2047444, and NIH R01 Grant LM013361-01A1. 

\section*{References}
\vspace{-1cm}
\bibliographystyle{abbrvnat}
\bibliography{/Volumes/flashdrive/research/bibliography}

\bigskip

\setlength{\parindent}{0cm}
\appendix

\section{Appendix}

\subsection{Two Derivations of the  Influence Function of the ATE}

\subsubsection{Integral Equation Approach}

To derive the influence function via the pathwise differentiability condition,  one needs to first calculate the derivative of the parameter
\begin{equation}
\frac{\partial}{\partial \epsilon} \psi(\Pb_\epsilon) \Bigm|_{\epsilon=0} \label{eq:ateder}
\end{equation}
on any smooth parametric submodel  $\Pb_\epsilon$ (satisfying $\Pb_0=\Pb$), set it equal to the inner product (i.e., covariance)
\begin{equation}
\int \varphi(z)\left( \frac{\partial}{\partial \epsilon} \log d\Pb_\epsilon(z) \right)\Bigm|_{\epsilon=0} \ d\Pb(z) \label{eq:ateinner}
\end{equation}
and solve this integral equation for the influence function $\varphi(z)$. In what follows we tackle these two tasks separately, for the average treatment effect parameter. \\

\subsubsection*{Derivative of parameter}

First we will compute the derivative of the parameter in \eqref{eq:ateder} for $\Pb_\epsilon$ a generic smooth parametric submodel satisfying $\Pb_0=\Pb$. Note that the submodel score decomposes as a sum
\begin{align*}
s_\epsilon(z) \equiv \frac{\partial}{\partial \epsilon} \log d\Pb_\epsilon(z)  &= \frac{\partial}{\partial \epsilon} \log d\Pb_\epsilon(y \mid x,a) + \frac{\partial}{\partial \epsilon} \log d\Pb_\epsilon(a \mid x) + \frac{\partial}{\partial \epsilon} \log d\Pb_\epsilon(x)  \\
&\equiv s_\epsilon(y \mid x,a) + s_\epsilon(a \mid x) + s_\epsilon(x) .
\end{align*}
Therefore, using the definition of the parameter, its derivative equals
\begin{align*}
\frac{\partial}{\partial \epsilon} \psi(\Pb_\epsilon) &= \frac{\partial}{\partial \epsilon} \int \int y \ d\Pb_\epsilon(y \mid x,1) \ d\Pb_\epsilon(x) \\
&= \int \int y \  \frac{\partial}{\partial \epsilon} d\Pb_\epsilon(y \mid x,1) \ d\Pb_\epsilon(x) +
		 \int \int y \  d\Pb_\epsilon(y \mid x,1) \ \frac{\partial}{\partial \epsilon} d\Pb_\epsilon(x) \\
&= \int \int y \  s_\epsilon(y \mid x, 1) \ d\Pb_\epsilon(y \mid x,1) \ d\Pb_\epsilon(x)  + \int \int y \  d\Pb_\epsilon(y \mid x,1) \ s_\epsilon(x) \ d\Pb_\epsilon(x) ,
\end{align*}
where in the third equality we used the derivative of a logarithm, i.e., that $\frac{\partial}{\partial \epsilon} \log f_\epsilon(z) = \frac{\partial}{\partial \epsilon} f_\epsilon(z)/f_\epsilon(z)$. Therefore the derivative of the ATE on a submodel at $\epsilon=0$ is given by
\begin{equation}
\frac{\partial}{\partial \epsilon} \psi(\Pb_\epsilon) \Bigm|_{\epsilon=0} = \int \int \Big\{  y \ s_0(y \mid x, 1) + \mu(x) s_0(x) \Big\}  \ d\Pb(y \mid x, 1) \ d\Pb(x) . \label{eq:atelhs}
\end{equation}
Now our task is to write this derivative in inner product form as in \eqref{eq:ateinner}. \\

\subsubsection*{Expressing as inner product}

Specifically here we aim to write the derivative \eqref{eq:atelhs} in the form
\begin{equation*}
\int \varphi(z) \Big\{ s_0(y \mid x,a) + s_0(a \mid x) + s_0(x)  \Big\} \ d\Pb(z) \label{eq:aterhs}
\end{equation*}
for a mean zero function $\varphi$, i.e., we need to solve the integral equation
\begin{equation}
\int \left\{  \int y \ s_0(y \mid x, 1) \ d\Pb(y \mid x, 1)  + \mu(x) s_0(x) \right\}  \ d\Pb(x) = \int \varphi(z)  s_0(z)   \ d\Pb(z) \label{eq:ateinteq}
\end{equation}
for $\varphi$. This does not have an obvious solution if the correct form for $\varphi$ is not already known (in which case one can just check that the equation holds). A potentially time-consuming and error-bound path is to keep conjecturing candidates for $\varphi$, until one works. \\

Instead it can be helpful to start by decomposing the influence function as $\varphi(z) = \varphi_y(y,x,a) + \varphi_a(a,x) + \varphi_x(x)$
where
\begin{align} \label{eq:ateifdecomp}
\int \varphi_y(y,x,a) \ d\Pb(y \mid x,a) &= 0 \nonumber \\
\int \varphi_a(a,x) \ d\Pb(a \mid x) &= 0 \\
\int \varphi_x(x) \ d\Pb(x) &= 0 . \nonumber
\end{align}
Note that this decomposition holds for any random variable $\varphi=\varphi(Z)$ by simply centering appropriately, i.e., defining $\varphi_y(Y,X,A)=\varphi - \E(\varphi \mid X,A)$, $\varphi_a(A,X) = \E(\varphi \mid X,A) - \E(\varphi \mid X)$, and $\varphi_x(X) = \E(\varphi \mid X) - \E(\varphi)$.  Importantly, with this decomposition,  the inner product on the right-hand side of \eqref{eq:ateinteq} simplifies to
\begin{align*}
\int  & \Big\{  \varphi_y(y,x,a) + \varphi_a(a,x) + \varphi_x(x) \Big\} \Big\{ s_0(y \mid x,a) + s_0(a \mid x) + s_0(x)  \Big\} \ d\Pb(z) \\
& = \int \Big\{ \varphi_y(y,x,a) s_0(y \mid x,a)  + \varphi_a(a,x) s_0(a \mid x) + \varphi_x(x) s_0(x) \Big\} \ d\Pb(z) 
\end{align*}
by virtue of the restrictions in \eqref{eq:ateifdecomp}, and the fact that $s_0(y \mid x,a)$, $s_0(a \mid x)$, and $s_0(x)$ are score functions and so similarly have conditional mean zero, i.e., 
\begin{align*}
\int s_0(y \mid x,a) \ d\Pb(y \mid x,a) &= 0  \\
\int s_0(a \mid x) \ d\Pb(a \mid x) &= 0 \\
\int s_0(x) \ d\Pb(x) &= 0 . 
\end{align*}
By using the decomposition \eqref{eq:ateifdecomp} we have essentially transformed our problem from solving one big integral equation to solving three smaller and easier integral equations. Specifically we now need to solve for $(\varphi_y,\varphi_a,\varphi_x)$ in 
\begin{align*}
& \int \int \Big\{  y \ s_0(y \mid x, 1) + \mu(x) s_0(x) \Big\}  \ d\Pb(y \mid x, 1) \ d\Pb(x) \\
& \hspace{.5in} = \int \Big\{ \varphi_y(y,x,a) s_0(y \mid x,a)  + \varphi_a(a,x) s_0(a \mid x) + \varphi_x(x) s_0(x) \Big\} \ d\Pb(z) 
\end{align*}
This immediately leads us to the choices $\varphi_a(a,x) = 0$ and $\varphi_x(x)=\mu(x) - \E\{\mu(X)\}=\mu(x)-\psi$. \\

The remaining task is to find $\varphi_y(y,x,a)$, which based on the above amounts to solving the integral equation  
\begin{equation*}
 \int \int y \ s_0(y \mid x, 1) \ d\Pb(y \mid x, 1) \ d\Pb(x) = \int  \varphi_y(y,x,a) s_0(y \mid x,a)  \ d\Pb(z) 
\end{equation*}
Note the derivative term on the left-hand side above equals
\begin{align*}
 \int \int y \ s_0(y \mid x, 1) \ d\Pb(y \mid x, 1) \ d\Pb(x) &= \int  \int \int y \ s_0(y \mid x, 1) \ d\Pb(y \mid x, 1) \ d\Pb(a \mid x) \ d\Pb(x) \\
 &= \int  \int \int \frac{ay}{\pi(x)} \ s_0(y \mid x, a) \ d\Pb(y \mid x, a) \ d\Pb(a \mid x) \ d\Pb(x)
\end{align*}
where in the first equality we simply introduced the treatment distribution $d\Pb(a \mid x)=a \pi(x) + (1-a)(1-\pi(x))$, and in the second we multiplied by the indicator $a$ to pick out the required $s_0(y \mid x,1)$ term, and then divided by $\pi(x)$ to cancel out the additional $\pi(x)$ term that appears due to averaging the indicator $a$. We note this logic is perhaps still a bit mysterious for those who have not used it before.  The term $ay/\pi(x)$ has the right mean, but needs to be centered so that it has conditional mean zero and thus satisfies the first line of \eqref{eq:ateifdecomp}. This leads to the choice
$$ \varphi_y(y,a,x) = \frac{ay}{\pi(x)} - \int \frac{ay}{\pi(x)} \ d\Pb(y \mid x,a) = \frac{a}{\pi(x)} \Big\{ y - \mu(x) \Big\} . $$
Combining with $\varphi_a(a,x)$ and $\varphi_x(x)$ above gives the overall influence function
\begin{equation}
\varphi(Z) = \frac{A}{\pi(X)} \Big\{ Y - \mu(X) \Big\} + \mu(X) - \psi .
\end{equation}
Since the work above shows that this satisfies the pathwise differentiability condition for any sufficiently smooth parametric submodel, it is the influence function. \\

\subsubsection{Gateaux Derivative Approach}

The Gateaux derivative approach can be viewed as a special case of the above, where one uses a particular choice of parametric submodel, for which the pathwise derivative is actually equal to the influence function, rather than an integral equation that needs to be solved. This is accomplished by using a submodel whose score is a point mass, as described in the main text. \\

One simple such submodel  is given by $\Pb^*_\epsilon(z) = (1-\epsilon) \Pb(z) + \epsilon \delta_{Z}$, where $\delta_{Z}$ is the Dirac measure at $z=Z$. Since  $z$ is discrete, we can just work with the mass function $p_\epsilon^*(z)=(1-\epsilon) p(z) + \epsilon \one(z=Z)$. 
First note that for the submodel $p_\epsilon^*(z)$ we have
\begin{align*}
p_\epsilon^*(y \mid x,a) &= \frac{p_\epsilon^*(z)}{p_\epsilon^*(a,x)} = \frac{(1-\epsilon) p(z) + \epsilon \one(z=Z)}{(1-\epsilon) p(a,x) + \epsilon \one(a=A,x=X)} \\
p_\epsilon^*(a \mid x) &= \frac{p_\epsilon^*(a,x)}{p_\epsilon^*(x)} = \frac{(1-\epsilon) p(a,x) + \epsilon \one(a=A,x=X)}{(1-\epsilon) p(x) + \epsilon \one(x=X)} \\
p_\epsilon^*(x) &= (1-\epsilon) p(x) + \epsilon \one(x=X)
\end{align*}
and
\begin{align}
\frac{\partial}{\partial \epsilon} p_\epsilon^*(y \mid x,a) \Bigm|_{\epsilon=0} &= \frac{ \one(z=Z) - p(z)}{(1-\epsilon) p(a,x) + \epsilon \one(a=A,x=X)}\Bigm|_{\epsilon=0} \nonumber \\
& \hspace{.5in} - p_\epsilon^*(y \mid x,a) \frac{\one(a=A,x=X)-p(a,x)}{(1-\epsilon) p(a,x) + \epsilon \one(a=A,x=X)} \Bigm|_{\epsilon=0} \nonumber \\
&=  \frac{ \one(z=Z) - p(z)}{ p(a,x)}  - p(y \mid x,a) \frac{\one(a=A,x=X)-p(a,x)}{ p(a,x)  } \nonumber \\
&= \one(a=A,x=X) \left\{ \frac{ \one(y=Y) - p(y \mid x,a)}{p(a,x)}    \right\} \label{eq:condensder}
\end{align}
where the first equality follows from the chain rule, and the rest by just rearranging. \\

Now we evaluate the parameter on the submodel, differentiate, and set $\epsilon=0$, which  gives 
\begin{align*}
\frac{\partial}{\partial \epsilon} \psi(p^*_\epsilon) \Bigm|_{\epsilon=0} &= \frac{\partial}{\partial \epsilon}  \sum_{x,y} y \ p_\epsilon^*(y \mid x,1) \ p_\epsilon^*(x) \Bigm|_{\epsilon=0}  \\
&=   \sum_{x,y} y \left\{ \frac{\partial}{\partial \epsilon} p_\epsilon^*(y \mid x,1) \ p_\epsilon^*(x) + p_\epsilon^*(y \mid x,1) \ \frac{\partial}{\partial \epsilon}  p_\epsilon^*(x) \right\} \Bigm|_{\epsilon=0}  \\
&=  \sum_{x,y} y \left[   \one(1=A,x=X) \left\{ \frac{ \one(y=Y) - p(y \mid x,1)}{p(1,x)}    \right\}p(x)
+ p(y \mid x,1) \Big\{ \one(x=X) - p(x) \Big\} \right] \\
&=  \frac{A}{\pi(X)} \Big\{ Y - \mu(X) \Big\} + \mu(X) - \psi 
\end{align*}
where the second equality follows by the chain rule, the third by substituting in the expression in \eqref{eq:condensder}, and the fourth rearranging. \\

Therefore the Gateaux derivative approach gives the same influence function as the more involved integral equation approach, though even the Gateaux derivative required more than a page of calculations, and some care with the submodel derivatives.
Note also that the influence function we arrived at is perfectly well-defined outside of the discrete setup, as long as the regression functions $\pi$ and $\mu$ are well-defined. \\
%%%%%%%%%%%%
%%%%%%%%%%%%
%%%%%%%%%%%%
%%%%%%%%%%%%
%%%%%%%%%%%%

\end{document}